\documentclass[journal,10pt]{IEEEtran}
\usepackage{pgfplots} 
\usepackage{comment}
\usepackage{amsfonts,enumitem}
\usepackage{textcomp}
\def\BibTeX{{\rm B\kern-.05em{\sc i\kern-.025em b}\kern-.08em
		T\kern-.1667em\lower.7ex\hbox{E}\kern-.125emX}}
\usepackage{amsmath,amssymb,multirow}
\usepackage{setspace}
\usepackage{listings}
\usepackage{cite}
\usepackage{stackengine}
\usepackage{acronym}
\usepackage{mathtools}
\usepackage{subfigure}
\usepackage{algorithm,algorithmic}
\usepackage{bm}

\pgfplotsset{
	tick label style={font=\small},
	label style={font=\small},
	legend style={font=\small}
}

\usepackage{enumitem}
\usepackage{mwe} 

\DeclareMathOperator*{\argmax}{arg\,max}
\DeclareMathOperator*{\argmin}{arg\,min}

\newcommand\abs[1]{\big|#1\big|}

\newcommand\norm[1]{\left\lVert #1\right\rVert}

\newcommand{\Mod}[1]{\ (\mathrm{mod}\ #1)}

\newcommand{\aB}{\abs{\mathcal{B}}}

\newcommand{\bs}{{\boldsymbol{s}}}
\newcommand{\bg}{{\boldsymbol{g}}}

\newcommand{\psdmax}{P} %\text{psd}_{\text{max}}}
\newcommand{\PSDmax}{P} %\text{PSD}_{\text{max}}}

\DeclarePairedDelimiter\floor{\lfloor}{\rfloor}
\makeatletter \renewcommand\d[1]{\ensuremath{%
		\;\mathrm{d}#1\@ifnextchar\d{\!}{}}}
\makeatother

\makeatletter
\newcommand*\rel@kern[1]{\kern#1\dimexpr\macc@kerna}
\newcommand*\widebar[1]{%
	\begingroup
	\def\mathaccent##1##2{%
		\rel@kern{0.8}%
		\overline{\rel@kern{-0.8}\macc@nucleus\rel@kern{0.2}}%
		\rel@kern{-0.2}%
	}%
	\macc@depth\@ne
	\let\math@bgroup\@empty \let\math@egroup\macc@set@skewchar
	\mathsurround\z@ \frozen@everymath{\mathgroup\macc@group\relax}%
	\macc@set@skewchar\relax
	\let\mathaccentV\macc@nested@a
	\macc@nested@a\relax111{#1}%
	\endgroup
}
\makeatother

\usepackage{amsthm}

\theoremstyle{remark}

\newtheoremstyle{mytheoremstyle} % name
{\topsep}                    % Space above
{\topsep}                    % Space below
{\upshape}                   % Body font
{.5em}                           % Indent amount
{\itshape}                   % Theorem head font
{.}                          % Punctuation after theorem head
{.5em}                       % Space after theorem head
{}  % Theorem head spec (can be left empty, meaning ‘normal’)

\theoremstyle{mytheoremstyle}

\newtheoremstyle{iremark}
{\topsep}   % ABOVESPACE
{\topsep}   % BELOWSPACE
{\upshape}  % BODYFONT
{0.2in}       % INDENT (empty value is the same as 0pt)
{\itshape}  % HEADFONT
{.}         % HEADPUNCT
{5pt plus 1pt minus 1pt} % HEADSPACE
{\thmname{#1}\thmnumber{ \itshape#2}\thmnote{ (#3)}} % CUSTOM-HEAD-SPEC

\theoremstyle{plain}
\newtheorem{prop}{Proposition}
\newtheorem{lemma}{Lemma}
\newtheorem{definition}{Definition}

\graphicspath{{./Figures/}}
\pgfplotsset{compat=1.18} 
\allowdisplaybreaks 

\newcommand{\dg}[1]{{\color{black}#1}}

\newcommand{\co}[1]{{\color{black}#1}}

\newcommand{\coold}[1]{{\color{black}#1}}

\begin{document}
    \title{Spectral Efficiency of Low Earth Orbit Satellite Constellations}
    \author{Cuneyd Ozturk,  
        Dongning Guo,  
        Randall A.~Berry,    
        and Michael L.~Honig 
        \thanks{C. Ozturk is with Aselsan Inc., Ankara, 06200, Turkey (E-mail: cuneydozturk@aselsan.com)
        D. Guo, R. Berry and M. L. Honig are with the Department of Electrical and Computer Engineering, Northwestern University, Evanston, IL, 60208 USA (E-mails: \{dguo, rberry, mhonig\}@northwestern.edu)}
        % ; rberry@northwestern.edu; mhonig@northwestern.edu)}
        \thanks{{This work was supported in part by the NSF under Grant No.~1910168 and SpectrumX, an NSF Spectrum Innovation Center under Grant No.~2132700.}}
    }
    \maketitle
    \begin{abstract}
        \co{
        This paper investigates the maximum achievable downlink spectral efficiency of low Earth orbit (LEO) satellite constellations. %In this context, 
        Spectral efficiency is defined \dg{here} as the total network sum rate per unit bandwidth per unit area of Earth's surface. To estimate an upper bound on %this metric, %we reduce the problem 
        \dg{spectral efficiency, the problem is reduced} to a %simplified
        single-channel network model, where all satellites and ground terminals operate over a common narrowband frequency channel. Within this %framework, 
        \dg{model}, %we propose and analyze 
        a regular benchmark configuration \dg{is proposed and analyzed, with satellites and terminals arranged} in %which satellites and terminals are arranged on 
        hexagonal lattices.
        %Numerical evaluations confirm that the spectral efficiency achieved under this configuration serves as a meaningful upper bound for more general multi-channel LEO networks, particularly when satellite-terminal associations are formed based on minimum-distance. Additional gains are demonstrated by modifying the 
        \dg{Numerical results validate that this configuration provides an upper bound on spectral efficiency for multi-channel LEO networks when satellite-terminal associations minimize the total squared link distance. %when satellite-terminal associations are based on minimum distance. 
        Further improvements are achieved by adjusting}
        association rules to prevent neighboring satellites from simultaneously serving terminals in the same region, highlighting the %importance
        \dg{critical role} of interference-aware association strategies.
       }
    \end{abstract}
    
    \section{Introduction}\label{sec:Intro}
        Satellite communications have recently experienced a significant wave of innovation, investment, and competition, with non-geostationary (NGSO) constellations in low Earth orbit (LEO) becoming a central focus. %This trend is 
        Driven by %technological advancements in 
        satellite miniaturization, reduced %tions in 
        launch costs, and improvements in communication protocols~\cite{li2022techno}, %Compared to traditional geostationary systems, 
        LEO constellations offer %key advantages such as
        lower latency, higher data throughput, and broader global coverage compared to geostationary systems~\cite{2019_Perez_Neira, 2017_LEO_IoT, 2021_Liu_LEOSat}. In the United States, %this momentum is reflected in the efforts of 
        over twenty companies %seeking authorization from the 
        have applied for
        Federal Communications Commission (FCC) authorization to deploy more than 70,000 LEO satellites across Ku-, Ka-, and V-bands~\cite{2022_Kriezis_USMarketAccessMegaConstNetworks, berry2024spectrum}.

        % A key distinction between terrestrial and satellite networks lies in how capacity scales with the number of access points. In
        {Unlike} terrestrial cellular systems, % increasing base station
        {where an increase in access point} density {enhances frequency reuse so that}
        % reduces the average distance between users and access points, enabling         spatial reuse.
        % As a result, total network throughput 
        {throughput} scales nearly linearly %with infrastructure density 
        in interference-limited regimes~\cite{Andrews16_Fundamental_Limits}, %. In contrast, satellite
        {LEO} networks face geometric constraints: The link distance between a satellite and its ground terminal %cannot fall below the satellite’s 
        is lower bounded by the orbital altitude, % Consequently, increasing the number of 
       and adding satellites leads to stronger aggregate interference, limiting spectral efficiency.

        Analyzing the spectral efficiency of LEO constellations is valuable for several key reasons. First, it enables an assessment of the capability of LEO systems to meet the growing global demand for reliable internet connectivity. Second, it offers insights into scalability and performance constraints, which can inform network design and deployment strategies. This, in turn, allows satellite service providers to make better informed investment decisions. Additionally, such analysis may contribute to the development of satellite communication policies and help shape future regulatory frameworks~\cite{berry2024spectrum, hazlett2023open}.
            
        Several prior studies have employed stochastic geometry to model satellite and terminal distributions as Poisson point processes or related frameworks~\cite{2020_Yastrebova_TheoreticalandSimulationBasedAnalysis, 2020_Okati_DLCoverageandRateNAnalysis, 2021_Hourani_AnAnalyticalApproach, 2021_Alouini_StochGeomLEOSatComm, 2022_Jung_SatelliteShadowedRician, 2022_Jia_AnalyticApproachUplinkMegaConst, 2022_Lee_ATractableApproach}. These works primarily focus on analyzing the performance of uplink and downlink LEO satellite networks in terms of key metrics such as coverage probability, average achievable rate, and outage probability. By leveraging stochastic geometry, the authors derive analytical expressions characterizing network performance under spatial randomness and validate their findings through numerical simulations.

        In this study, we adopt spectral efficiency measured in bits per second per Hz per square kilometer as the primary performance metric. It is defined as the total downlink data rate divided by the product of the available bandwidth and the surface area of Earth covered by the satellite constellation. Our goal is to characterize how spectral efficiency scales with satellite density and to identify the key trade-offs that emerge in densely deployed LEO systems.

        This work makes the following contributions and distinctions from the existing literature~\cite{2020_Yastrebova_TheoreticalandSimulationBasedAnalysis, 2020_Okati_DLCoverageandRateNAnalysis, 2021_Hourani_AnAnalyticalApproach, 2021_Alouini_StochGeomLEOSatComm, 2022_Jung_SatelliteShadowedRician, 2022_Jia_AnalyticApproachUplinkMegaConst, 2022_Lee_ATractableApproach}:

        \begin{itemize}
             \item \textcolor{black}{We establish that, for any downlink LEO satellite network, there exists a
             virtual \emph{single-channel network} in which each satellite serves a single terminal over a shared narrowband channel whose spectral efficiency upper bounds that of the original network (see Sec.~\ref{sec:ReductiontoSingleChan}).} 
             
             %We show that for any given downlink LEO satellite network, there exists a virtual single-channel network where each satellite serves a single terminal over a shared narrowband channel whose spectral efficiency upper bounds that of the original network (see Sec.~\ref{sec:ReductiontoSingleChan}). 
             
             \item \textcolor{black}{ We introduce a stylized \emph{regular configuration} for the single-channel
            network (Sec.~\ref{sec:RegularConfig}), for which the derived spectral efficiency is a closed-form function of the inter-satellite spacing (equivalently, the satellite density) and the transmit power, with the remaining system parameters held fixed. Through simulations, we further show that, when satellite-terminal pairs are associated so as to minimize the total squared link distance, the spectral efficiency of the regular configuration provides a meaningful upper bound on that of randomly generated multi-channel networks.}

            \item \textcolor{black}{ While most prior works adopt a minimum-distance association between satellites and terminals, we show that this heuristic becomes suboptimal in the high-density regime. We therefore propose an alternative beam-to-terminal assignment that accounts for nearest-neighbor interference in addition to link distance, and that can significantly increase the network sum capacity at high densities (see Sec.~\ref{sec:AssociationMax}).}

        \end{itemize}
        
        The remainder of the paper is organized as follows. Sec.~\ref{sec:SystemModel} introduces the system model.
        %for a generic LEO satellite network. 
        Sec.~\ref{sec:SingleChan} defines the single-channel network and establishes its relevance as a bound. Sec.~\ref{sec:RegularConfig} presents the regular configuration \textcolor{black}{for the single-channel network}. %obtaining the maximum 
        %bounding the spectral efficiency under distance-based associations. 
        Sec.~\ref{sec:AssociationMax} explores alternative association strategies. % beyond total distance minimization. 
        Sec.~\ref{sec:NumRes} %provides 
        presents numerical results %validating the analysis 
        and Sec.~\ref{sec:Conc} concludes the paper and discusses directions for future work.

%\dg{HOW ABOUT SAYING ``(ground) terminal'' ONCE AND THEN JUST terminal INSTEAD OF terminal?}

    \section{System Model} \label{sec:SystemModel}

        The \dg{average} downlink spectral efficiency, expressed in units of [bits/s/Hz/km$^2$], is determined by several \dg{major} factors:
        \begin{enumerate}
            \item the total number of satellites and terminals, their spatial distribution, and their association,
            % \item the spatial distribution of both satellites and terminals,
             % \item the association strategy between satellites and terminals,
             \item the transmission schedule in time and frequency,
            \item %the 
            satellite transmission power, % levels of the satellites,
            \item antenna beamwidths of both satellites and terminals,        
            \item antenna boresight orientations.      
        \end{enumerate}
         %Taking these key factors into consideration, we proceed to 
         We provide a comprehensive model %description 
         of the satellite network, detailing these parameters, and
         % used in our analysis. Furthermore, we 
         express the downlink spectral efficiency as a function thereof %these parameters 
         to understand their impact on system performance. 

        \subsection{Geometry}
             Terminals are positioned on the Earth's surface, which is modeled as a perfect sphere with a radius of $r_e = 6378$~km. Satellites are placed at an altitude of $h$~km, lying on the surface of a larger concentric sphere with a radius of $r_e + h$~km\footnote{\co{Our framework and analysis are readily applicable to satellite networks operating at multiple orbital altitudes. However, for the sake of notational simplicity and to facilitate a clearer understanding of the core problem, we focus on a single-altitude configuration throughout this work.}}. The coordinate system is centered at the point $[0, 0, 0]$, representing the center of the Earth and simultaneously the center of both spheres.

        \subsection{Satellite-Terminal Associations}
    
            \begin{figure}
                \centering
                \includegraphics[width = \columnwidth]{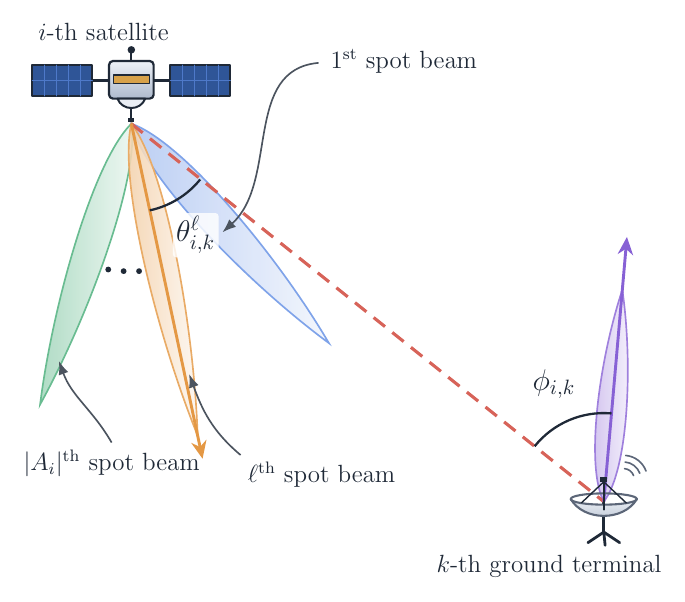}
                \caption{Illustration of %off-axis angles for the 
                satellite spot beams and a terminal \dg{beam}.}
                \label{fig:sat_gs_off-axis}
            \end{figure}

            We assume that each satellite can provide service to multiple terminals by utilizing multiple spot beams as illustrated in Fig.~\ref{fig:sat_gs_off-axis}. We exclude the scenario where a terminal is being served by multiple satellites. In addition, we focus on single-hop satellite-terminal links.
            
            Let $F(\cdot)$ denote the association between the satellites and the terminals, i.e., terminal $k$ is served by satellite $F(k)$.  Also, satellite $i$ serves the set of the terminals $A_i$ so that $F(j) = i$ for $j\in A_i$. Satellite $i$ controls $\abs{A_i}$ independent spot beams where $\abs{\cdot}$ is the cardinality of the set. In addition, $s_k\in\left\{1, \ldots, \abs{A_{F(k)}}\right\}$ denotes the serving spot beam index for ground terminal $k$. Received signals from other spot beams are treated as interference. As a baseline, we consider the association rule that minimizes the total squared link distance between the associated satellite-terminal pairs  (see \cite{2020_Yastrebova_TheoreticalandSimulationBasedAnalysis, 2020_Okati_DLCoverageandRateNAnalysis, 2022_Okati_NonHomeg}).

         \subsection{System Constraints: Time, Frequency, and Power}\label{subsec:TimeandFreq}

            % To derive an upper bound on spectral efficiency, we observe that 
            \dg{A}lthough a mega-constellation is continuously in motion, the time-averaged spectral efficiency cannot exceed \dg{that achieved} in its most favorable configuration. Accordingly, we consider \dg{short} time slots 
            %of sufficiently short duration 
            (on the order of milliseconds) during which the satellite positions can be assumed \dg{fixed}.
            % to remain effectively unchanged.
            
            Let $\mathcal{B}$ represent the set of frequency bands allocated for downlink satellite communication, the total bandwidth of which is denoted as            $\abs{\mathcal{B}}$.
            This set comprises            $M$ 
            non-overlapping subbands, 
            denoted             $\mathcal{B}_1, \dots, \mathcal{B}_M$, such that their union is equal to
            $\mathcal{B}$.
           Let $S_i\subseteq \{1, \dots, M\}$ denote the subband indices allocated to satellite $i$, and          $G_k\subseteq \{1, \dots, M\}$ denote those %the subband indices 
           allocated to terminal $k$, which must satisfy           $G_k\subseteq S_{F(k)}$.
           
           During every time slot, the $\ell^{\text{th}}$ spot beam of satellite $i$ must satisfy a constraint on its power spectral density (PSD) 
          \begin{align} \label{eq:Pf}
                \coold{P_i^{\ell}(f)} \leq  \PSDmax , \, 
                \forall \, f\in\mathcal{B}
            \end{align}
           and a constraint on its total transmit power:
           \begin{align} \label{eq:Pmax}
                \int_{\mathcal{B}} P_i^{\ell}(f)\, df \leq P_{\text{max}}
            \end{align}
            where $P_i^{\ell}(f)$ is the PSD of the $\ell^\mathrm{th}$ spot beam for satellite $i$. 

                When links are in close %physical 
                proximity and \dg{cause} %generate 
                significant interference, assigning them to orthogonal subbands can be beneficial. However, %the 
                subband allocation %problem 
                is %inherently 
                a combinatorial \dg{problem} and computationally challenging.
                %, particularly for arbitrarily distributed terminals, as scalable optimal solutions are unlikely. 
                To address this, we \dg{partition} %assume that 
                the Earth's surface %is partitioned 
                into regions, each assigned a dedicated subband, \dg{akin} %similar 
                to frequency planning in cellular networks. 
                A terminal’s Voronoi region is %defined as 
                the set of all points on the surface %that are 
                closer to that terminal than to any other. If the center of a predefined subband region lies within a terminal’s Voronoi region, that subband is allocated to the terminal. Each satellite then uses the union of all subbands \dg{allocated to terminals it serves}.
                % %associated with its connected terminals. 
                In Sec.~\ref{sec:NumRes}, we \dg{illustrate this strategy with a hexagonal frequency reuse pattern as a practical example.}

        \subsection{Propagation Model}\label{subsec:Propagation}

            The distance between satellite $i$ and terminal $k$ is denoted by $d_{i,k}$. %Based on this, 
            The path-loss between satellite $i$ and terminal $k$ is modeled as $d_{i,k}^{-\alpha}$, where $\alpha\ge 2$ is the path loss exponent. Referring to Fig.~\ref{fig:sat_gs_off-axis}, let $\theta_{i,k}^{\ell}$ denote the off-axis angle between the 
            boresight of the $\ell^{\mathrm{th}}$ spot beam for satellite $i$ and the direction from satellite $i$ to terminal $k$.   Similarly,  $\phi_{i,k}$ denotes the off-axis angle between the            boresight of terminal $k$ and the direction from terminal $k$ to satellite $i$. 
            
            The antenna pattern of each satellite spot beam is represented by the function $w_s(\cdot)$, while $w_g(\cdot)$ represents the antenna pattern of the terminals. The link gain between the $\ell^{\mathrm{th}}$ spot beam of satellite $i$ and terminal $k$ is modeled as
            \begin{align}
                \Omega_{i, k}^{\ell}(f)\triangleq\psi(f)d_{i,k}^{-\alpha} w_{s}(\theta_{i,k}^{\ell}) w_{g}(\phi_{i,k}),
                \label{eq:Omega}
            \end{align}
           where $\psi(f)$ captures the gain's frequency dependence.
           % \footnote{In this study, perfect Doppler compensation is assumed when considering an upper bound on the maximum spectral efficiency. As noted in \cite{Yeh_2024_Doppler_Compensation}, Doppler shifts in LEO communication systems can be effectively compensated under high SNR conditions, across a range of Doppler shift magnitudes and channel environments, including varying satellite elevation angles.

           % To maintain notational simplicity and enhance the clarity of the core analysis, fading effects are initially omitted. However, in Sec.~\ref{sec:NumRes}, we incorporate fading for all satellite–terminal pairs using the Shadowed Rician fading model. The corresponding average spectral efficiency is then evaluated by averaging over multiple fading realizations.} 
           \dg{Recall that $s_k$ denotes the spot beam index for terminal $k$.} The SINR at ground terminal $k$ is then  
            \begin{align}
                \mathrm{SINR}_k(f) = \frac{P^{s_k}_{F(k)}(f) \Omega_{F(k), k}^{s_k}(f)}{\sum_{(i, \ell)\neq (F(k), s_k)} P_{i}^{\ell}(f)\ \Omega_{i, k}^{\ell}(f) + N_0},
                \label{eq:SINR_kf}
            \end{align}
            where $N_0$ denotes the single-sided PSD of the AWGN (additive white Gaussian noise) at each terminal. 
            The spectral efficiency objective for the \dg{entire} satellite network %\co{over the entire earth} 
            is
            \begin{align}\label{eq:spec_eff_general}
                R = \frac{1}{4\pi r_e^2 \abs{\mathcal{B}}}\sum_{k} \sum_{m\in G_k}\int_{\mathcal{B}_m} \log_2\left(1 +  \mathrm{SINR}_k(f)  \right)\, df.
            \end{align}
         
         In this work, our goal is to estimate the maximum achievable spectral efficiency as a function of satellite and terminal beamwidths and satellite transmission power constraints. Contrary to the typical assumption in the literature~\cite{2022_Okati_NonHomeg, 2020_Yastrebova_TheoreticalandSimulationBasedAnalysis}, the satellites' boresights are not assumed to be directed at Earth's center. Instead, both the satellites and terminals are allowed flexibility in their \dg{beam} directions, which can lead to improved spectral efficiencies. While we focus on downlink transmissions, a similar  analysis can also be applied to the uplink scenario.
         
          In this study, perfect Doppler compensation is assumed. This is a reasonable idealization for LEO networks when the goal is an upper bound on the spectral efficiency, because the Doppler is dominated by the known motion of the satellites and is deterministic once the terminal position and satellite ephemeris are available. In standardized NTN systems the terminal obtains its position from GNSS and the satellite position and velocity from the broadcast ephemeris. Therefore, the geometric Doppler can be predicted and removed~\cite{3gpp_38821}, with the residual remaining small at high SNR across a wide range of Doppler magnitudes, channel conditions, and elevation angles~\cite{Yeh_2024_Doppler_Compensation}.

         % In this study, perfect Doppler compensation is assumed% in deriving an upper bound on the spectral efficiency.
         % This is a reasonable idealization for LEO networks as far as derivation of an upper bound on the spectral efficiency is concerned.  The Doppler is dominated by the known motion of the satellites and is deterministic once the terminal position and satellite ephemeris are available. In standardized NTN systems the terminal obtains its position from GNSS and the satellite position and velocity from the broadcast ephemeris, so the geometric Doppler can be predicted and removed~\cite{3gpp_38821}, with the residual remaining small across a wide range of Doppler magnitudes and elevation angles~\cite{Yeh_2024_Doppler_Compensation}.

        % In this study, perfect Doppler compensation is assumed when considering an upper bound on the maximum spectral efficiency. As noted in \cite{Yeh_2024_Doppler_Compensation}, Doppler shifts in LEO communication systems can be effectively compensated under high SNR conditions, across a range of Doppler shift magnitudes and channel environments, including varying satellite elevation angles. In standardized NTN systems, the geometric Doppler shift is deterministic and can be compensated, since a terminal can determine its own position from GNSS and obtain the satellite position and velocity from the ephemeris broadcast by the network~\cite{3gpp_38821}. As the ground terminals considered here are at known locations, the residual Doppler after compensation is negligible, supporting the assumption of perfect compensation.
        
        To maintain notational simplicity and enhance the clarity of the core analysis, fading effects are initially omitted. However, in Sec.~\ref{sec:NumRes}, we incorporate fading for all satellite–terminal pairs using the Shadowed Rician fading model. The corresponding average spectral efficiency is then evaluated by averaging over multiple fading realizations.

\section{Single-Channel Network}\label{sec:SingleChan}
    
        In this section, a virtual satellite network, referred to as the single-channel network, is introduced to derive an upper bound on the spectral efficiency of the satellite network described in Sec.~\ref{sec:SystemModel}. Although this network is purely hypothetical, it offers valuable information-theoretic insights for estimating the maximum spectral efficiency. 

\subsection{Reduction to Single-Channel Network}
\label{sec:ReductiontoSingleChan}

        \begin{definition}
        A         single-channel network \dg{is one where:} %if}
            \begin{enumerate}
                \item The association between satellites and terminals is one-to-one, with each satellite serving its terminal through a single beam;
                %Each terminal is served by a unique satellite %;
                % \item Each satellite serves a unique terminal 
                %\dg{using} %with 
                %a single beam;
                \item All transmissions use the same \dg{narrow-band} channel. %  over a narrow subband.
            \end{enumerate}
        \end{definition}

       % \dg{We introduce} the single-channel network \dg{for the following reason:}
            \begin{prop}\label{Prop:single_channel_reduction}
                For any satellite network described in Sec.~\ref{sec:SystemModel}, there exists a single-channel network, subject to the same PSD constraint, whose spectral efficiency %serves as an 
                upper bounds %for 
                that of the original network.
            \end{prop}
            \begin{proof}
                First, without limiting the number of terminals, we can replace terminal $k$ with $\abs{G_k}$ co-located terminals,  each served by a different subband allocated to satellite $F(k)$. This operation does not affect \eqref{eq:spec_eff_general}. Applying this procedure to all terminals gives a network in which each terminal is served by a single subband.
                
                Second, without limiting the number of satellites, we can replace satellite $i$ with $\abs{S_i}\times\abs{A_i}$ satellites at the same position, each with a single spot beam assigned a single subband and subject to the same power constraints~\eqref{eq:Pf} and~\eqref{eq:Pmax}. This is transparent to the receivers since their received signals do not change. Hence, as far as maximizing the spectral efficiency is concerned, it suffices to consider one dedicated terminal for each satellite. 
                
                The entire network is partitioned into $M$ subnetworks, each consisting of distinct satellite–terminal pairs operating on a separate subband. Since the overall spectral efficiency is the average across all subbands, it is necessarily upper bounded by the subband achieving the highest spectral efficiency.
            \end{proof}
            
            %It is important to note that 
            Proposition~\ref{Prop:single_channel_reduction} remains valid even when fading and/or frequency-dependent gain variations across satellite–terminal pairs are incorporated into the model. In the presence of fading, the spectral efficiency expression in \eqref{eq:spec_eff_general} should be replaced by its expected value %, i.e., the average spectral efficiency computed 
            over %the 
            fading realizations.

            \dg{Additionally, although the analysis assumes uniform satellite altitudes, the single-channel reduction is independent of this assumption and remains valid with varying altitudes.}
            % It is also important to note that, although the analysis assumes all satellites are positioned at the same altitude, the single-channel reduction is not tied to this geometric constraint. The reduction remains valid even when satellite altitudes vary across the constellation.
            
            As a consequence of Proposition~\ref{Prop:single_channel_reduction}, the maximum spectral efficiency achieved among all single-channel networks provides an upper bound for any satellite network. \dg{Thus, we focus on single-channel networks to establish this bound.} %Consequently, the class of single-channel networks will be considered to establish this bound. 
    
             Under the assumptions of a single spot beam and narrowband operation, the spot beam index $\ell$ and the frequency dependence $f$ are omitted. % from the notation. In particular, 
             The link gain between satellite $i$ and terminal $k$ is expressed as: %$\Omega_{i,k}$, where
            \begin{align}
                \Omega_{i,k} =  d^{-\alpha}_{i, k}w_{s}(\theta_{i,k}) w_{g}(\phi_{i,k}),
            \end{align}
            where the common frequency-flat gain $\psi$ in~\eqref{eq:Omega}  has been absorbed into the noise by defining $\sigma^2 \triangleq N_0/\psi$. 
            %For the subsequent analysis, 
            %By normalizing the noise PSD with respect to $\psi$, %we assume $\psi = 1$ without loss of generality, and 
            %we denote the resulting noise level of the single-channel model by $\sigma^2 \triangleq N_0/\psi$. 
            % By normalizing the noise PSD with respect to $\psi$, we %can, without loss of generality 
            % assume $\psi$ = 1 without loss of generality. %Furthermore, 
            Under narrowband \dg{operation}, % assumption, 
            only the PSD constraint~\eqref{eq:Pf} \dg{applies, simplifying} %remains active. Specifically, for satellite $i$, the PSD constraint simplifies 
            to $P_i\le P$.
            
            %In addition, 
            Since the number of satellites \dg{equals the number of} terminals in this virtual network, their locations are represented by  $(\bs_{i})_{i=1}^{N}$  and $(\bg_{k})_{k=1}^{N}$, respectively, where $N$ denotes the total number. % of satellites (terminals).  
            The %associations are described by 
            \dg{corresponding} $F(\cdot)$ %, which 
            is a permutation of $\{1,\dots,N\}$ %chosen to 
            \dg{that} minimizes the total squared distance:
            \begin{align}\label{eq:assoccost}
                \sum_{k=1}^{N} \norm{\bs_{F(k)}-\bg_{k}}^2
            \end{align}
            where $\norm{\cdot}$ denotes Euclidean norm.  The Hungarian algorithm~\cite{1955_Kuhn_HungarianMethod} can be used to determine $F(\cdot)$. We then \dg{reindex} the %indices of the 
            terminals so that $F(k) = k$.

        \subsection{Optimal Power Allocation under \dg{Symmetry}}\label{sec:Opt_Power_Allocation}

    Let $\mathcal{I}_k$ denote the interfering satellites for terminal $k$ (all satellites above its horizon except the serving one). The single-channel SINR is
    \begin{align}\label{eq:SINR_k}
    \mathrm{SINR}_k=\frac{P_k\,\Omega_{k,k}}{\sum_{i\in\mathcal{I}_k}P_i\,\Omega_{i,k}+\sigma^2}.
    \end{align}
  The following proposition gives conditions \dg{under} which transmitting at full power is locally optimal.

    % We now identify conditions under which full-power transmission is optimal in a single-channel network. Let $\mathbf P=(P_1,\dots,P_N)$ collect the transmit powers and define the network sum spectral efficiency
    % \begin{align}
    %     \Phi(\mathbf P)\triangleq\sum_{k=1}^{N}\log_2\!\left(1+\mathrm{SINR}_k\right),
    % \end{align}
    % with $\mathrm{SINR}_k$ as in~\eqref{eq:SINR_k}.

\begin{prop}\label{Prop:fullpower}
    The allocation $P_k=\psdmax$ for all $k$ is a strict local maximizer of the spectral efficiency under the following symmetric conditions:
    \begin{enumerate}
        \item The direct-link gains are equal, $\Omega_{k,k}=\Omega_{1,1}$ for all $k$, and
        \item There exists a common value $\bar I$ such that, for every $k$,
        \begin{align}\label{eq:reciprocity}
            \sum_{i\in\mathcal I_k}\Omega_{i,k}=\!\!\sum_{m:\,k\in\mathcal I_m}\!\!\Omega_{k,m}=\bar I.
        \end{align}
    \end{enumerate}
    % Then:
    % \begin{enumerate}
    %     \item[(a)]  $P_k=\psdmax$ for all $k$ maximizes $\Phi$ over common-power allocations.
    %     \item[(b)] The full-power allocation $\mathbf P=\psdmax\mathbf 1$ is a strict local maximizer of $\Phi$ over $[0,\psdmax]^N$ 
        
    %     % with
    %     % \begin{align}\label{eq:fullpower_grad}
    %     %     \frac{\partial \Phi}{\partial P_j}\bigg|_{\mathbf P=\psdmax\mathbf 1}=\frac{\Omega_{1,1}\,\sigma^2}{\log 2\;N\,T}>0,\quad\forall j,
    %     % \end{align}
    %     % where $N\triangleq\psdmax\bar I+\sigma^2$ and $T\triangleq N+\psdmax\Omega_{1,1}$.
    % \end{enumerate}
\end{prop}
    The proof of Proposition~\ref{Prop:fullpower} is provided in Appendix~\ref{sec:proofprop2}.

    We emphasize that Proposition~\ref{Prop:fullpower} is a \emph{local} statement. Global optimality of full power additionally requires sufficiently weak interference. In the high-density limit the dominant interferers approach the desired-signal strength, and selectively silencing transmitters can then be beneficial. Accordingly, we adopt full-power transmission as a benchmark.

        \subsection{Planar Approximation}\label{subsec:Flatplaneapprox}
            In scenarios where interference at a given terminal is primarily caused by nearby satellites within a localized region, the satellite positions can be approximated as lying on a plane. This planar approximation becomes increasingly accurate %when the 
            \dg{for dense} satellite networks % is sufficiently dense and the 
            \dg{with narrow} beamwidths \dg{for} both satellites and terminals.  Accordingly, we adopt a simplified model in which satellites and terminals are assumed to reside in parallel planes separated by a distance of $h$~km. 

            %In particular,
            For the planar approximation, we project satellite locations to the two-dimensional plane $\{[x, y, r_e+h] \mid x, y \in\mathbb{R}\}$. Similarly, we project terminal locations to $\{[x, y, r_e] \mid x, y \in\mathbb{R}\}$. 
             For example, if a satellite is located at $[x, y, z]$ according to the spherical model, its projected location is given by 
            \begin{align}
            \dg{                [x, y, z] \rightarrow [x (r_e +h) z^{-1}, \, y(r_e +h)z^{-1}, \, r_e +h]. }
                % [x, y, z] \rightarrow \Bigg[x \left(\frac{r_e +h}{z}\right), y\left(\frac{r_e +h}{z}\right), r_e +h\Bigg]. 
            \end{align}

            % \co{Consider a terminal located at $[0, 0, r_e]$. Any satellite whose zenith angle is smaller than $\theta_f$ lies in the field-of-view of the terminal, where
            % \begin{align}
            %     \theta_f\triangleq \cos^{-1}\left(\frac{r_e}{r_e +h}\right).
            % \end{align}
            % The location of a satellite lying on the field-of-view of the terminal located at $[0, 0, r_e]$ can be expressed as
            % \begin{align}
            %     (r_e+h) [\sin\zeta\cos\varphi~~ \sin\zeta\sin\varphi~~\cos\zeta]
            % \end{align}
            % for $0\le \varphi \le 2\pi$ and $0\le \zeta \le \theta_f$. }

            \begin{figure}
                \centering
                \includegraphics[width = \columnwidth]{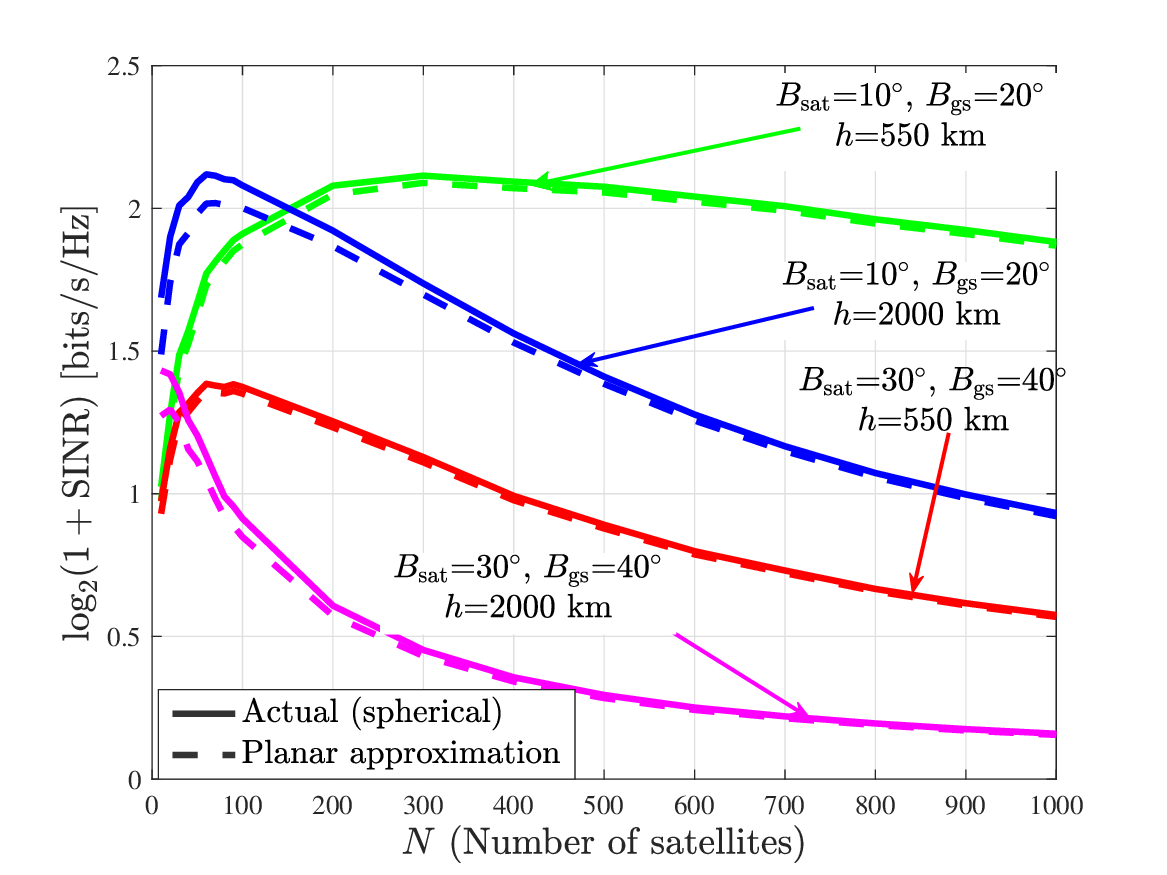}
                \caption{The individual rate [bits/s/Hz] at the terminal located at $[0, 0, r_e]$ versus the number of  satellites within the field-of-view ($N$) for the BPP model and planar approximation when $h\in \{550, 2000\}$~km.
                }
                \label{fig:flatvsBPP}
            \end{figure}

            We %compare 
            \dg{next demonstrate that the planar approximation is quite accurate for parameters of interest. Suppose}
            % averages of $\log_2(1 + \mathrm{SINR})$ over satellite and terminal locations at the fixed terminal location $[0, 0, r_e]$ for both the BPP model and planar approximation when $h\in\{550, 2000\}$~km. 
            %In particular, 
            $N$ satellites in the field-of-view of the terminal are generated according to the BPP.  
            Any satellite whose zenith angle is smaller than $\theta_f$ lies in the field-of-view of the terminal, where
            \begin{align}
                \theta_f\triangleq \cos^{-1}\left({r_e}/(r_e +h)\right).
                % \theta_f\triangleq \cos^{-1}\left(\frac{r_e}{r_e +h}\right).
            \end{align}
            The location of a satellite lying on the field-of-view of the terminal located at $[0, 0, r_e]$ can be expressed as
            \begin{align}
                (r_e+h) [\sin\zeta\cos\varphi~~ \sin\zeta\sin\varphi~~\cos\zeta]
            \end{align}
            for $0\le \varphi \le 2\pi$ and $0\le \zeta \le \theta_f$. From~\cite[Prop.~1]{2022_Alouini_EvaluatingAccuracyofStochGeomLEO}, $\varphi$ \dg{is uniformly distributed on} %\sim\mathcal{U}
            $[0, 2\pi]$ and the cumulative distribution function of $\zeta$ is given as
            \begin{align}
                F_{\zeta}(\theta) = \frac{1- \cos \theta}{1-\cos \theta_f},\;\; 0\leq \theta \leq \theta_f.
            \end{align}
            Similarly, we generate $(N-1)$ additional terminals on the Earth's surface according to the BPP, in addition to a terminal located at $[0, 0, r_e]$. The zenith angles of these terminals are restricted to the interval $[0, \theta_f]$, while their azimuth angles are uniformly distributed over $[0, 2\pi]$. Associations between satellites and terminals are established by minimizing the total distance metric defined in \eqref{eq:assoccost}. The %antenna 
            boresights of each satellite–terminal pair are assumed to be aligned.
            
            In this example, all satellites transmit at a power spectral density (PSD) level of  $\psdmax$. We evaluate two configurations for the satellite and terminal beamwidths, specifically $(B_{\mathrm{sat}}, B_{\mathrm{gs}}) = (10^\circ, 20^\circ)$ and $(30^\circ, 40^\circ)$, where $B_{\mathrm{sat}}$ and $B_{\mathrm{gs}}$ are the boresight-to-first-null angles\footnote{ The narrower setting ($B_{\mathrm{sat}} = 10^\circ$) corresponds to a half-power beamwidth of $\approx 8.4^\circ$, comparable to the $8.83^\circ$ half-power ($3$\,dB) beamwidth of the 3GPP NTN LEO ``Set-2'' configuration~\cite{3gpp_38821, 3gpp_beam_pattern}, while the wider setting is conservative.}.
            % We evaluate two configurations for the satellite and terminal beamwidths, specifically $(B_{\mathrm{sat}}, B_{\mathrm{gs}}) = (10^\circ, 20^\circ)$ and $(30^\circ, 40^\circ)$\textcolor{black}{, where $B_{\mathrm{sat}}$ and $B_{\mathrm{gs}}$ are the first-null beamwidths. The narrower setting yields a satellite half-power beamwidth comparable to standardized 3GPP NTN LEO spot beams~\cite{3gpp_TSG_RAN_WG4_meeting, 3gpp_beam_pattern}, while the wider setting is conservative.}
            Antenna patterns for both %satellites and terminals 
            are generated according to the model %described 
            in Sec.~\ref{subsec:AntennaPattern}. The SNR of the serving link is set to $8~$dB.

            % \textcolor{blue}{
            % We evaluate two configurations for the satellite and terminal beamwidths, specifically $(B_{\mathrm{sat}}, B_{\mathrm{gs}}) = (10^\circ, 20^\circ)$ and $(30^\circ, 40^\circ)$, where $B_{\mathrm{sat}}$ and $B_{\mathrm{gs}}$ denote first-null beamwidths. These settings are representative of the antenna directivities adopted for LEO systems in the 3GPP NTN studies~\cite{3gpp_TSG_RAN_WG4_meeting, 3gpp_beam_pattern}. In particular, the narrower satellite beamwidth corresponds to a half-power beamwidth of about $4^\circ$, comparable to the standardized LEO spot-beam widths. The broader setting is deliberately conservative, since wider beams reduce directivity and increase inter-beam interference, yielding more conservative spectral-efficiency estimates.}
            
            Fig.~\ref{fig:flatvsBPP} \dg{plots the} averages of $\log_2(1 + \mathrm{SINR})$ over satellite and terminal locations at the fixed terminal location $[0, 0, r_e]$ for both the BPP model and planar approximation when $h\in\{550, 2000\}$~km. The results indicate that the planar model \dg{is a close approximation}. %is quite accurate for parameters of interest. 
            Hence all subsequent analysis is carried out using the planar model.  As a result, we assume that the satellites and terminals extend to the infinite horizon ($N$ is infinite), so that their density becomes the parameter of interest.

            In summary, to estimate the maximum spectral efficiency of the satellite network in Sec.~\ref{sec:SystemModel}, we proceed as follows:
            \begin{enumerate}
                \item Reduce the problem to determine the maximum spectral efficiency of a single-channel network (Proposition~\ref{Prop:single_channel_reduction}),
                \item Derive a set of sufficient conditions under which full-power transmission is optimal in the single-channel network (Proposition~\ref{Prop:fullpower}),
                \item Approximate the spherical geometry with a planar model, numerically \dg{validate its accuracy for relevant} % demonstrated that this approximation is highly accurate for the 
                satellite altitudes. % of interest. 
            \end{enumerate}
            
        Next, we consider a \textit{regular configuration} for deployment of  satellites and terminals in the single-channel network.

    \section{Regular Configuration}
    \label{sec:RegularConfig}
        % Referring to Fig. \ref{fig:reg-placement}, {\em regular placement} applied to the planar approximation
        % means that the  satellites and terminals are deployed across hexagonal lattices. To formally define this, we will index the satellite and terminals by an ordered pair $(i,j)$ with $i = j \Mod2$. 

    \begin{definition}
    \label{def:reg_placement}
        Let $\Delta>0$ denote \dg{the %expected 
        distance between any satellite and its nearest neighbor, hereafter} referred to as inter-satellite spacing. We say the satellites are regularly placed if their locations are
           \begin{align} \label{eq:sijreg}
           \bs_{i, j}^{\text{reg}}\triangleq
           [ i \Delta/2,\, j \Delta \sqrt{3}/2,\, r_e + h ],
           % \Bigg[ i \frac{\Delta}{2},\, j \frac{\Delta \sqrt{3}}{2},\, r_e + h \Bigg],
           \quad
           i=j \Mod2.      
           \end{align}
            and the terminals are regularly placed if their locations are
            \begin{align} \label{eq:gijreg}
               \bg_{i, j}^{\text{reg}} \triangleq
               [  i \Delta/2,\, j \Delta \sqrt{3}/2,\, r_e ],
               % \Bigg[  i \frac{\Delta}{2},\, j \frac{\Delta \sqrt{3}}{2},\, r_e \Bigg],
               \quad
               i=j \Mod2. 
            \end{align}
         \dg{Here a satellite or terminal is indexed by a pair of integers $(i,j)$, which are both either odd or even. Moreover, we call it a {\em regular configuration} if every terminal $(i,j)$ is associated with satellite $(i,j)$, with their antennas pointing directly to each other.}
    	\end{definition}
        
        \begin{figure} 
    		\centering
    		\includegraphics[width = \columnwidth]{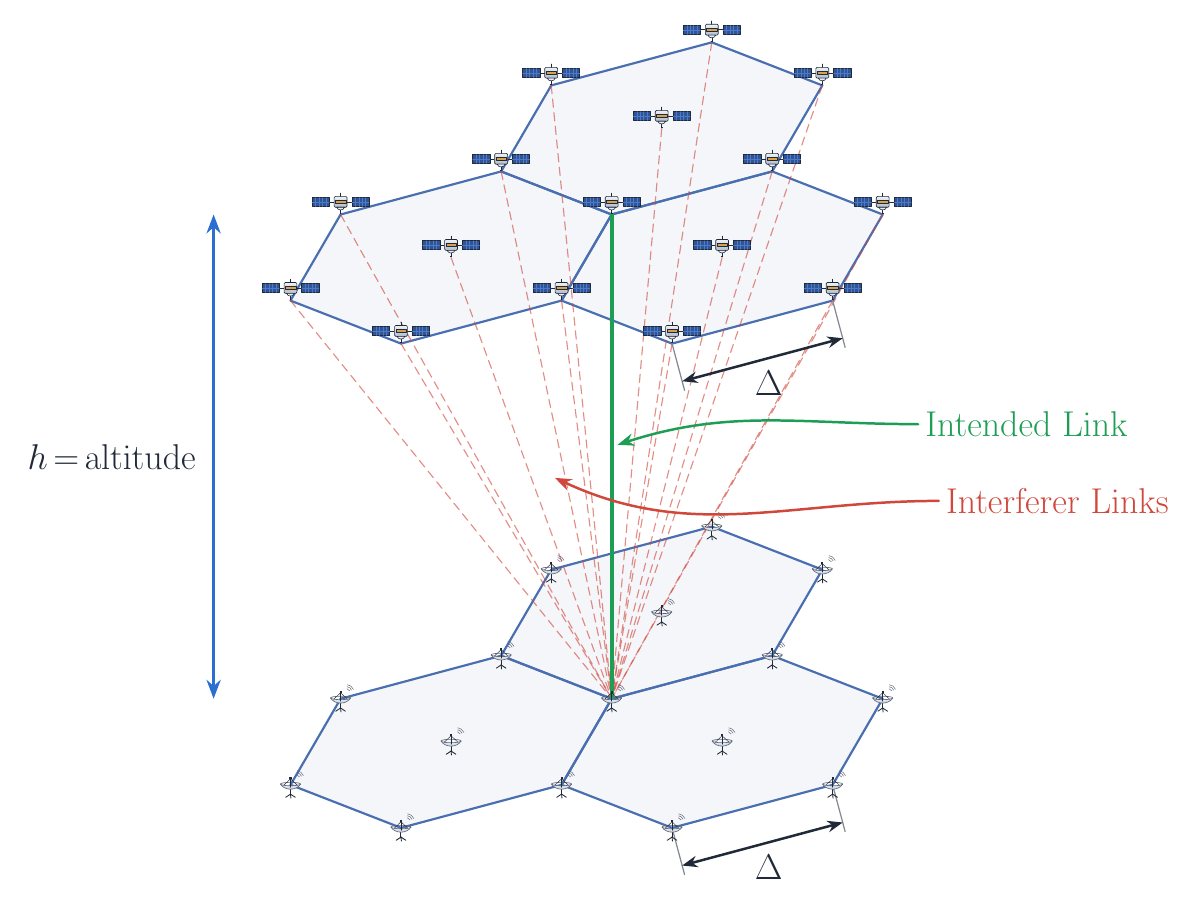}
    		\caption{Regular configuration of satellites and terminals in the planar model. The serving satellite is directly above the associated terminal.}
    		\label{fig:reg-placement}
    	\end{figure}	
        
        % A {\em regular configuration} consists of regular placements of both satellites and terminals where satellite $\bs_{i,j}^{\text{reg}}$ serves terminal $\bg_{i, j}^{\text{reg}}$ and their antennas point directly to each other. 
        
       % The regular configuration\dg{, illustrated in Fig.~\ref{fig:reg-placement},} minimizes the length of each direct link, thereby maximizing the received SNR at every terminal. It also maximizes the minimum distance between any pair of satellites, which in turn provides an upper bound on the interference that can occur between any pair of links.

       The regular configuration\dg{, illustrated in Fig.~\ref{fig:reg-placement},}  minimizes the length of every direct link (each equal to the altitude $h$), thereby maximizing the received SNR at every terminal, and it maximizes the minimum inter-satellite distance, so the
        dominant (nearest) interferers are pushed as far away as the density allows.
        With the rapidly decaying antenna gains considered here, this tends to suppress
        the aggregate interference. We verify numerically in Sec.~VI that the regular
        configuration outperforms randomly placed (BPP) networks under minimum-squared-distance association.
        
        %yields the highest spectral efficiency among distance-based
        %configurations.

        \subsection{Spectral Efficiency}
        \label{subsec:TheoResults}

         Conditions~1 and~2 in Proposition~\ref{Prop:fullpower} are met by the regular configuration. All direct links have gain $h^{-\alpha}$, and by its translational and up-down symmetry every terminal's incoming interference and every satellite's outgoing interference equal a common value $\bar I$. Hence full power is a strict local maximizer in the regular configuration. We assume that every satellite transmits at the maximum PSD $\psdmax$. 
    
        By the translational symmetry of the lattice, all terminals see the same geometry and therefore share a common SINR. It thus suffices to analyze a reference terminal positioned at $[0, 0, r_e]$, served by the satellite directly above it. Specializing the single-channel SINR~\eqref{eq:SINR_k} to this terminal, the direct link has length $h$, while the interferers are the lattice satellites
            \begin{align}\label{eq:Iset}
                \mathcal{I} \triangleq \left\{(i, j) \mid i, j \in \mathbb{Z}, \, i = j \Mod{2} \right\} \setminus \{(0, 0)\},
            \end{align}
         each seen at slant range
        $D^{\mathrm{reg}}_{i,j}(\Delta)\triangleq
        \bigl\lVert \bs^{\mathrm{reg}}_{i,j}-\bg^{\mathrm{reg}}_{0,0}\bigr\rVert$,
        %=\sqrt{(i\Delta/2)^2+(j\Delta\sqrt{3}/2)^2+h^2}$,
        so that $D^{\mathrm{reg}}_{0,0}(\Delta)=h$. Since each satellite points directly downward and each terminal directly upward, the satellite-side and terminal-side off-axis angles coincide, taking the common value
        \begin{align}
        \theta^{\mathrm{reg}}_{i,j}(\Delta)
        = \cos^{-1}\!\left(h\big/D^{\mathrm{reg}}_{i,j}(\Delta)\right).
        \label{eq:reg_angle}
        \end{align}
        % so that $D^{\mathrm{reg}}_{0,0}(\Delta)=h$, and at off-axis angle
        % \begin{align}
        % \theta^{\mathrm{reg}}_{i,j}(\Delta)
        % = \cos^{-1}\!\left(h\big/D^{\mathrm{reg}}_{i,j}(\Delta)\right).
        % \label{eq:reg_angle}
        % \end{align}
        The regular placement~\eqref{eq:sijreg}--\eqref{eq:gijreg} forms a hexagonal lattice with nearest-neighbor distance $\Delta$, so each terminal occupies a fundamental cell of area $\Delta^2\sqrt{3}/2$. The resulting spectral efficiency is
        \begin{align}\label{eq:throughputreg}
        R^{\mathrm{reg}}(P,\Delta)
        = \frac{2}{\Delta^2\sqrt{3}}\,
        \log_2\!\left(1+\frac{\gamma(P)}{\eta(P,\Delta)+1}\right),
        \end{align}
        where
        \begin{align}
        \gamma(P) &\triangleq P\,\sigma^{-2} h^{-\alpha}, \label{eq:gamma}\\
        \eta(P,\Delta) &\triangleq \frac{P}{\sigma^2}
        \sum_{(i,j)\in\mathcal{I}}
        \bigl(D^{\mathrm{reg}}_{i,j}(\Delta)\bigr)^{-\alpha}\,
        w_s\!\left(\theta^{\mathrm{reg}}_{i,j}(\Delta)\right)
        w_g\!\left(\theta^{\mathrm{reg}}_{i,j}(\Delta)\right).
        \label{eq:eta}
        \end{align}
        The following proposition characterizes the behavior of the spectral efficiency in the regular configuration as the satellite density increases.
            
            \begin{prop}\label{Prop:reginfty}
                As $\Delta \to 0$, the spectral efficiency $R^{\text{reg}}(\psdmax, \Delta)$ converges to a finite, strictly positive constant that depends on the transmission power $\psdmax$, noise power $\sigma^2$, altitude $h$, path loss exponent $\alpha$, and the beamwidths of both satellite and terminal antennas.
            \end{prop}
            The proof of this result is provided in Appendix~\ref{sec:proofprop3}. Proposition~\ref{Prop:reginfty} reveals that spectral efficiency does not grow indefinitely as the inter-satellite spacing $\Delta$ decreases. Instead, it saturates in the high-density regime, implying that the maximum achievable spectral efficiency occurs at a finite satellite density. Although an analytical expression for the optimal spacing $\Delta$ is difficult to extract directly from~\eqref{eq:throughputreg}, the optimization can be performed efficiently through numerical methods, since the objective function depends only on a single variable.

            Specifically, in the high-density regime, the spectral efficiency for both random and regular configurations can be approximated by
            \begin{align}\label{eq:Rcont}
                R^{\text{cont}}(\psdmax, \Delta) \triangleq  \frac{2}{\Delta^2 \sqrt{3}} \log_2\left(1 + \frac{P h^{-\alpha}}{\sigma^2 + I^{\text{cont}}(\psdmax, \Delta)}\right),
            \end{align}
            where the continuous interference approximation $I^{\text{cont}}(\psdmax, \Delta)$ is given by
            \begin{align}
                I^{\text{cont}}(\psdmax, \Delta) \triangleq \frac{P}{\Delta^2 \sqrt{3}/2} \int_{-\infty}^{\infty} \int_{-\infty}^{\infty} I(x, y) \, dx \, dy,
            \end{align}
            with the integrand defined as
            \begin{align}
                I(x, y) \triangleq \left(x^2 + y^2 + h^2\right)^{-\alpha/2} w_s\left( \kappa(x, y) \right) w_g\left( \kappa(x, y) \right),
            \end{align}
            and the angle function $\kappa(x, y)$ given by
            \begin{align}
                \kappa(x, y) \triangleq \cos^{-1}\left(h/{\sqrt{x^2 + y^2 + h^2}}\right).
                % \kappa(x, y) \triangleq \cos^{-1}\left(\frac{h}{\sqrt{x^2 + y^2 + h^2}}\right).
            \end{align}
           In Sec.~\ref{sec:NumRes}, the spectral efficiency of the regular configuration in the high-density regime is compared with the continuous approximation provided in~\eqref{eq:Rcont}.

    \subsection{Constrained Optimality of Regular Configuration}\label{subsec:OptPlacementGroundStats}
        
        % \mh{Proving the optimality of the regular configuration under a distance-based association rule presents significant analytical challenges. Here we instead provide an optimality result %\dg{offering} 
        % under a constrained scenario where the network consists of regularly spaced satellites. We then identify the terminal placements that maximize spectral efficiency. Those placements correspond to the regular configuration.}

        Proving the optimality of the regular configuration under a distance-based association rule presents significant analytical challenges. Here we instead provide an optimality result %\dg{offering} 
        under a constrained scenario where the network consists of regularly spaced satellites. We then identify the terminal placements that maximize spectral efficiency. Those placements correspond to the regular configuration.
        
        In this setting, each satellite transmits at full power, denoted by $\psdmax$, \dg{with} its antenna  oriented %to point 
        directly downward. Similarly, all terminals are %assumed to be 
        oriented directly upward. %As a result, 
        \dg{Thus,} the boresight axes of both satellite and terminal antennas are aligned perpendicular to the ground plane.
        Under these conditions, we show that spectral efficiency is maximized when each terminal is positioned directly beneath its associated satellite.
      
       Let $\bg = [x, y, r_e]$ denote the position of a terminal. For any satellite indexed by $(i, j)$ satisfying $i = j \Mod{2}$, define the Euclidean distance between the satellite and the terminal as
        \[
        D_{i,j}(x, y) \triangleq \norm{\bs_{i,j}^{\text{reg}} - \bg}.
        \]
        The combined attenuation due to the satellite and terminal beam patterns is defined as
        \[
        W_{i,j}(x, y) \triangleq w_s(\theta_{i,j}(x, y)) \, w_g(\theta_{i,j}(x, y)),
        \]
        where the off-axis angle $\theta_{i,j}(x, y)$ is given by
        \begin{align}
            \theta_{i,j}(x, y) = \cos^{-1} \left( h/{D_{i,j}(x, y)} \right).
            % \theta_{i,j}(x, y) = \cos^{-1} \left( \frac{h}{D_{i,j}(x, y)} \right).
        \end{align}
        Without loss of generality, we assume that the terminal at $\bg$ is served by the satellite located at $[0, 0, r_e + h]$. Under this association, the signal-to-interference-plus-noise ratio (SINR) at the terminal is expressed as
        \begin{align}
            \mathrm{SINR}(x, y) = \frac{D_{0,0}^{-\alpha}(x, y) \, W_{0,0}(x, y)}{
            \sum_{(i, j) \in \mathcal{I}} D_{i,j}^{-\alpha}(x, y) \, W_{i,j}(x, y) + \sigma^2 / \psdmax},
        \end{align}
        where $\mathcal{I}$ is the interference index set defined in~\eqref{eq:Iset}.
        % \[
        % \mathcal{I} \triangleq \left\{(i, j) \mid i, j \in \mathbb{Z}, \, i = j \Mod{2} \right\} \setminus \{(0, 0)\}.
        % \]
        
        \begin{prop}\label{Prop:regopt}
        The SINR function $\mathrm{SINR}(x, y)$ is maximized at $(x, y) = (0, 0)$ if the following conditions are satisfied:
        \begin{enumerate}
            \item The product of antenna gains, $w_s(\theta) w_g(\theta)$, is a non-increasing function of $\theta \in [0, \pi/2]$.
            \item The ratio $W_{i,j}(x, y) / W_{0,0}(x, y)$ is convex in $x$ and $y$ over the region $-\Delta/2 \leq x \leq \Delta/2$, $-\Delta \sqrt{3}/2 \leq y \leq \Delta \sqrt{3}/2$, for any $(i, j)$ such that $i = j \Mod{2}$.
        \end{enumerate}
        \end{prop}
        
        The proof of Proposition~\ref{Prop:regopt} is provided in Appendix~\ref{sec:proofprop4}.
        
        When the combined beam attenuation $W_{i,j}(x, y)$ satisfies the conditions stated in Proposition~\ref{Prop:regopt}, the optimal placement of the terminal served by the satellite at $[0, 0, r_e + h]$ is directly beneath it, at $[0, 0, r_e]$. By symmetry, this implies that all other terminals should likewise be positioned directly below their associated satellites in order to maximize SINR.
        
        In practice, antenna gain patterns \dg{may not be} %are not always 
        monotonic due to %the presence of 
        sidelobes and nulls, \dg{so} %. Therefore, 
        Proposition~\ref{Prop:regopt} may not \dg{always apply}. %be universally applicable. 
        However, monotonic envelope functions \dg{bounding} %that bound 
        the actual %antenna 
        patterns %from above or below 
        can \dg{lead to} %be used to         derive %meaningful         upper or lower 
        bounds on %the received 
        interference and %consequently on 
        spectral efficiency. We also \dg{numerically} verify %through numerical evaluation 
        that the second condition in Proposition~\ref{Prop:regopt} holds for the antenna %beam 
        patterns %considered 
        in Sec.~\ref{sec:NumRes}.

\section{Shuffling Associations Can Increase %Spectral 
Efficiency}\label{sec:AssociationMax}

    Up to this point, satellites have been associated with terminals by minimizing a distance-based metric, as in prior works~\cite{2020_Yastrebova_TheoreticalandSimulationBasedAnalysis, 2020_Okati_DLCoverageandRateNAnalysis, 2022_Okati_NonHomeg}. % We have also introduced the regular configuration as a promising candidate for maximizing spectral efficiency in the single-channel network model when the distance-based association rule is employed.
    In this section, we \dg{continue to assume a single-channel network and} demonstrate that spectral efficiency can be further enhanced by modifying the association strategy and adjusting the antenna look directions. % Throughout this section, we continue to assume a single-channel network setting.
    
    \dg{Antenna gains decay rapidly as the receiver moves away from the boresight. In particular,} %we define 
    two links %as strongly interfering 
    interfere strongly if both their satellites and their terminals are close neighbors. Consider \dg{the two different associations} in Fig.~\ref{fig:look-dir}.
    %, which compares two %different satellite-to-terminal 
    % association schemes. %In such a case, each link acts as a strong interferer to the other.
    % In the association shown on the left of Fig.~\ref{fig:look-dir}, 
    \dg{On the left,} every link has one or two strong interferers; \dg{on the right}, the association eliminates strong interference. Notably, 
    \dg{switching associations by slightly reorienting antennas increases the distance by a small factor, reducing direct-link signal strength by a few dB while dramatically reducing interference.}
    % with small} %when the     inter-satellite spacing, \dg{switching associations in this manner only marginally reduces direct link signal strength.}
    % is small, the signal strength of the direct links is only marginally reduced by switching from the left-hand to the right-hand association.
    
    \begin{figure}
        \centering
        \includegraphics[width = \columnwidth]{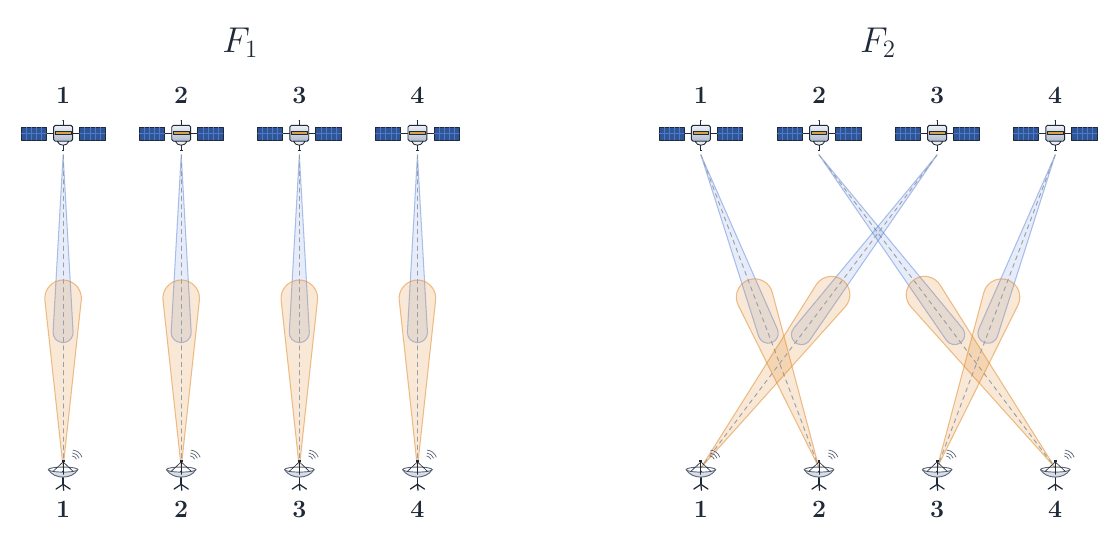}
        \caption{Comparison of satellite–terminal associations. The left configuration minimizes the satellite-to-terminal distance, while the right configuration reduces interference at each terminal.}
        \label{fig:look-dir}
    \end{figure}

    To mitigate strong interference, we adopt a simple association principle. Satellites that lie close to one another are assigned to terminals that are far apart, so that strongly interfering links do not coexist. % To \dg{mitigate} strong interference, we propose a heuristic for association: If two satellites are within distance $\Delta_s$, assign \dg{them} to terminals separated by at least $\Delta_g$. Here, $\Delta_s$ and $\Delta_g$ are design parameters. % that govern the spatial separation between interfering entities.   More 
    % Formally, for any indices $i$ and $k$, we require the association function $F(\cdot)$ to satisfy the following condition:
    % \begin{align}
    %     \norm{\bs_{F(i)} - \bs_{F(k)}} \leq \Delta_s \implies \norm{\bg_i - \bg_k} > \Delta_g.
    %     \label{eq:association_rule}
    % \end{align}
    \dg{Next, we present} %, we introduce 
    an algorithm %designed to improve 
    \dg{to enhance} the spectral efficiency of networks with regularly placed satellites and terminals \dg{using an}
    % by leveraging this 
    interference-aware association. % strategy.

    % \subsection{Shuffling Algorithm}

        We begin with the regular configuration \dg{in one dimension}, %along a single spatial dimension, 
        as %illustrated 
        in Fig.~\ref{fig:look-dir}, \dg{restricting} %where 
        the %total 
        number of satellites %(or equivalently, terminals) is restricted 
        to $2^j$ for %some 
        natural number $j$. Given a satellite index $k \in \{1, \ldots, 2^j\}$, % we define a
        \dg{the} one-dimensional association to terminals \dg{is} % using the function
        \begin{align} \label{eq:fmn}
            f^{(j)}(k) = 
            \begin{cases}
                (k + 2^j + 1)/2, & \text{if } k \text{ is odd}, \\
                k/2, & \text{if } k \text{ is even}.
                % \frac{k + 2^j + 1}{2}, & \text{if } k \text{ is odd}, \\
                % \frac{k}{2}, & \text{if } k \text{ is even}.
            \end{cases}
        \end{align}
        This ensures adjacent satellites associate with terminals separated by at least $2^{j-1}$ positions.
        %To extend this mapping to a larger domain where 
        \dg{For} $k > 2^j$, we apply periodic extension:
        \begin{align} \label{eq:fmn_extend}
            f^{(j)}(k) = 2^j \left\lfloor 2^{-j}(k-1) \right\rfloor + f^{(j)}\left( k - 2^j \left\lfloor 2^{-j}(k-1) \right\rfloor \right).
            % f^{(j)}(k) = 2^j \left\lfloor \frac{k-1}{2^j} \right\rfloor + f^{(j)}\left(k - 2^j \left\lfloor \frac{k-1}{2^j} \right\rfloor \right).
        \end{align}
        
        For a two-dimensional regular grid, % of satellites and terminals, 
        let $F_x(\cdot)$ and $F_y(\cdot)$ denote associations along $x$ and $y$ axes, applying one-dimensional shuffling independently. We can further iterate to enhance interference suppression. % by applying multiple rounds of shuffling.
        Parameters include shuffling block sizes $D_x = 2^{n_x}$ and $D_y = 2^{n_y}$ for integers $n_x, n_y \ge 1$, and rounds $\ell_x$ and $\ell_y$ satisfying $\ell_x \le n_x - 1$ and $\ell_y \le n_y - 1$.
        %In the case of a single shuffling round, the mapping reduces to 
        \dg{For one round,} $F_x(m) = f^{(n_x)}(m)$. In general, %for arbitrary $\ell_x$ and $n_x$, the multi-round shuffling strategy is expressed as
        \begin{align}
            F_x(m) = f^{(n_x - \ell_x + 1)} \big( f^{(n_x - \ell_x + 2)} \big( \cdots f^{(n_x)}(m) \cdots \big) \big).
        \end{align}
        
        Figures~\ref{fig:shuffling_round_1} and~\ref{fig:shuffling_round_2} \dg{show} one and two rounds of shuffling in one dimension with $D = 8$. More rounds (e.g., from $\ell = 1$ to $\ell = 2$) increase minimum distance between serving satellites and interfering terminals but may increase satellite-terminal distance. In high-density scenarios, \dg{interference reduction may outweigh path loss increase, improving} spectral efficiency.
        
        To extend to two dimensions, recall the regular locations $\bs_{i,j}^{\text{reg}}$ and   $\bg_{i,j}^{\text{reg}}$ from~\eqref{eq:sijreg} and~\eqref{eq:gijreg}. Shuffling applies independently along axes. Let $F^{\text{shuffle}}(i, j)$ denote the shuffled indices of the terminal associated with $\bs_{i,j}^{\text{reg}}$:
        \begin{align}\label{eq:shuffle2D}
            F^{\text{shuffle}}(i, j) = \left( 2 F_x\left( (i - q)/{2} \right) + r, \, F_y(j) \right),
            % F^{\text{shuffle}}(i, j) = \left( 2 F_x\left( \frac{i - q}{2} \right) + r, \, F_y(j) \right),
        \end{align}
        where $q \in \{0, 1\}$ is $i$ modulo 2, and $r \in \{0, 1\}$ is $F_y(j)$ modulo 2.

    \begin{figure}
        \begin{center}
            \subfigure[]{
            \label{fig:shuffling_round_1}
            \includegraphics[width = \columnwidth]{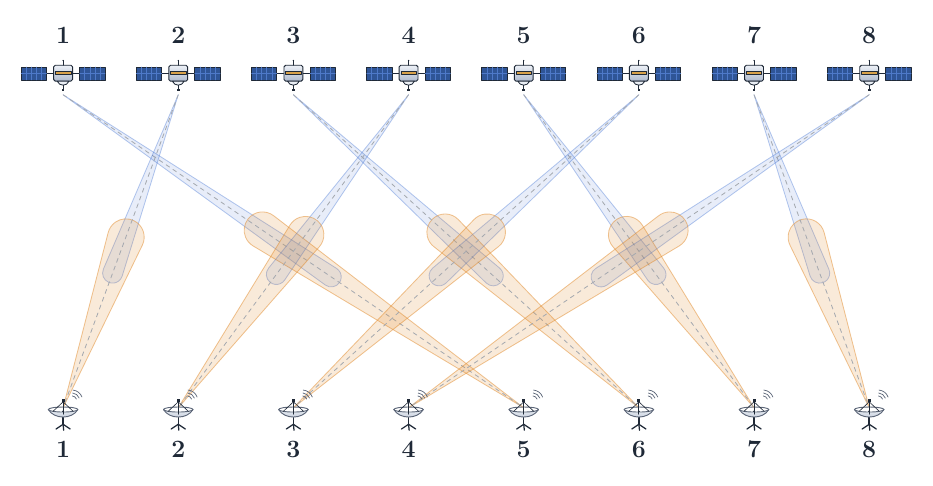}
            }
            \subfigure[]{
            \label{fig:shuffling_round_2}
            \includegraphics[width = \columnwidth]{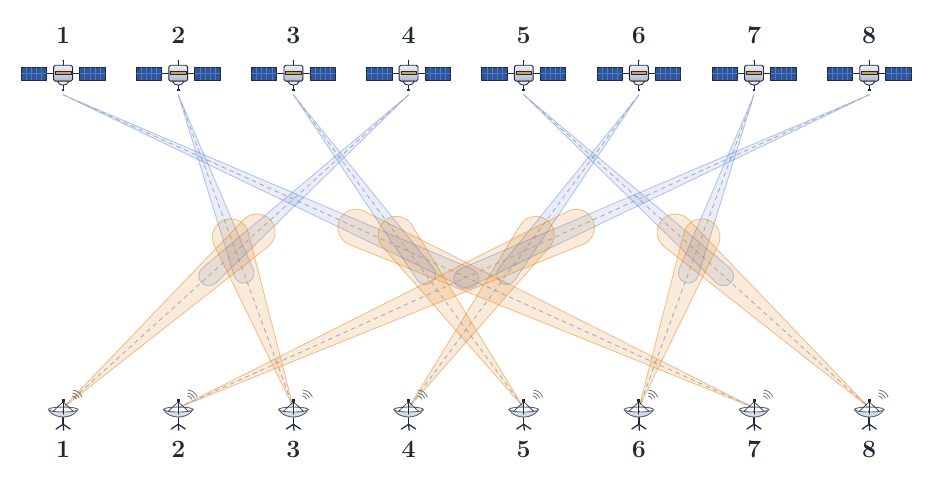}
            }
        \end{center}
        \caption{Shuffling strategies when $D = 8$ and $\ell=1$ in (a) and $\ell=2$ in (b).}
        %\label{fig:}
    \end{figure}

    Once shuffling parameters $(\ell_x, \ell_y, D_x, D_y)$ are specified, the association is computed in closed form \dg{via~}\eqref{eq:shuffle2D}. \dg{Thus,} computational complexity arises solely from searching possible  parameter combinations \dg{for the optimal outcome}.

\section{Numerical Results}\label{sec:NumRes}
    \subsection{Antenna Patterns and System Parameters}\label{subsec:AntennaPattern}
        \co{Very-small aperture terminals (VSATs) are assumed for the terminals. Following the proposal by THALES in \cite{3gpp_TSG_RAN_WG4_meeting}, the same antenna beam pattern is used for both terminals and satellites.  In particular, given off-axis angle $\theta$ and boresight-to-first-null angle $B$, the antenna pattern is expressed as~\cite{3gpp_beam_pattern, 3gpp_beam_pattern_2, kim2024downlink} 
        \begin{align}\label{eq:beam-andrews}
        	g(B, \theta) =  \begin{cases}
        		1, & \text{if } \theta = 0^\circ \\
        		4 \abs{\frac{J_1 (K\sin\theta)}{K\sin\theta}}^2, & \text{for } 0^\circ < \abs{\theta}\leq 90^\circ,
        	\end{cases}
        \end{align}
        where  $J_1(\cdot)$ is the first-order Bessel function of the first kind and $K$ determines the beamwidth.  }

        That is, $B = \sin^{-1}\left(3.8317 /K\right)$; equivalently, the conventional (null-to-null) first-null beamwidth equals $2B$, and the half-power beamwidth is $2\sin^{-1}\left(1.6163/K\right)$, all determined solely by $K$. Given satellite and terminal beamwidths $B_{\rm{sat}}$ and $B_{\rm{gs}}$, the associated antenna patterns are
        \begin{align}
        	w_s(\theta) \triangleq g(B_{\rm{sat}},\theta), \,  w_g(\theta) \triangleq g(B_{\rm{gs}},\theta).\label{eq:beam_pattern_final}
        \end{align}
        For the examples in this section we assume that $h = 550$ km, the path loss exponent $\alpha = 2.5$ and all path attenuations are frequency-flat. Throughout this section, spectral efficiency is reported in bits/s/Hz per $1000~\mathrm{km}^2$, i.e., $10^3$ times the per-km$^2$ value, for readability.

    \subsection{Single Channel Network: Random versus Regular Configuration}\label{subsec:SimRes}
        \co{The primary objective of this section is to evaluate the performance of the regular configuration for the single-channel network when distance-based association is employed. Specifically, the regular configuration is compared against random configurations, in which satellites and terminals are independently and uniformly distributed over two parallel planar surfaces, as described in Sec.~\ref{subsec:Flatplaneapprox}. Associations between satellites and terminals are established by minimizing the squared-distance metric in \eqref{eq:assoccost}. Additionally, the look directions of each associated satellite–terminal pair are assumed to be aligned.}
            
        \begin{figure}
            \centering
            \includegraphics[width = \columnwidth]{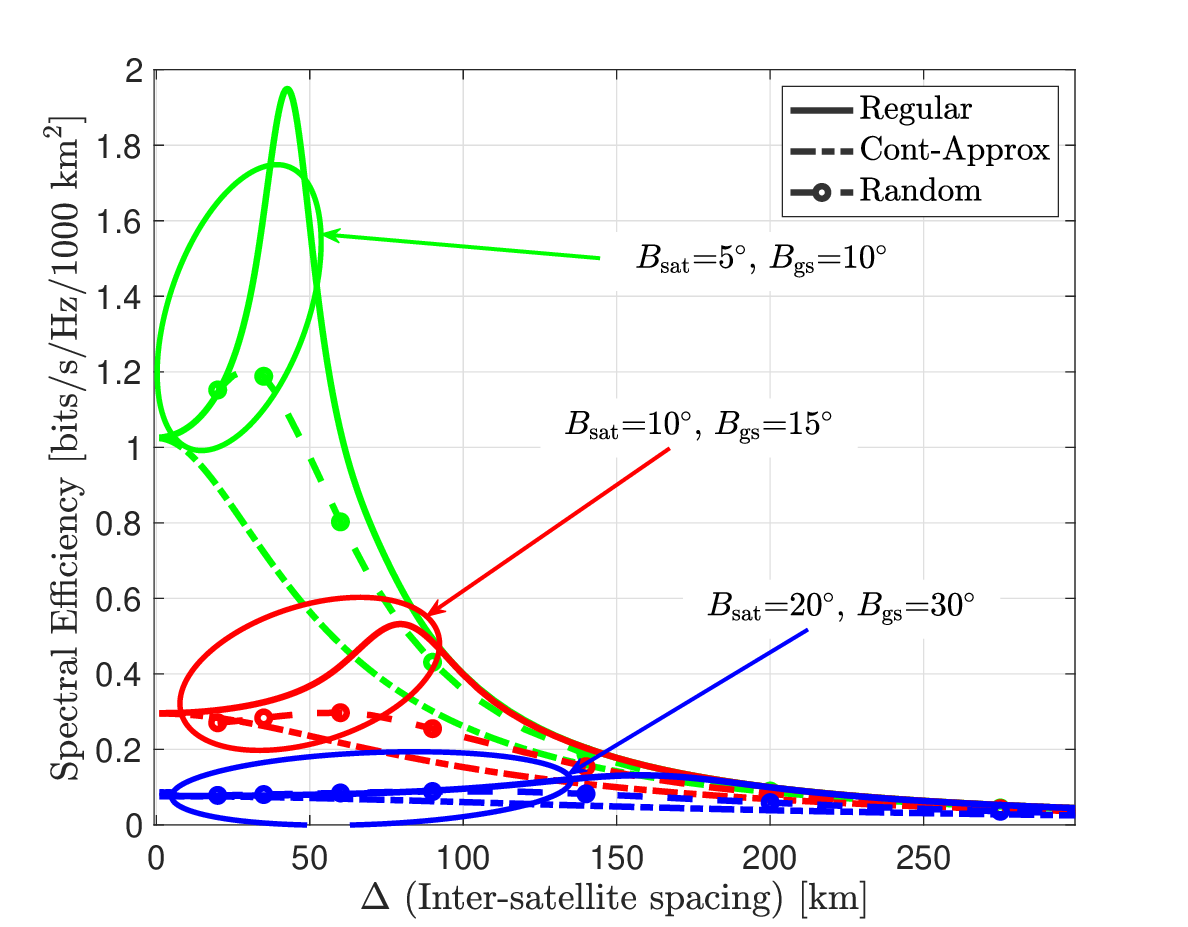}
            \caption{Spectral efficiency versus inter-satellite spacing with randomly and regularly placed satellites and terminals along with the continuous approximation for the single-channel network. Results are shown for three sets of beamwidths  $(B_{\rm{sat}},B_{\rm{gs}})\in\{(5^\circ, 10^\circ), (10^\circ, 15^\circ), (20^\circ, 30^\circ)\}$ and $\psdmax h^{-\alpha}/\sigma^2 = 10~$dB.}
        \label{fig:single_channel_random_reg_various_Bsat_Bgs}
        \end{figure}
        \begin{figure}
            \centering
            \includegraphics[width = \columnwidth]{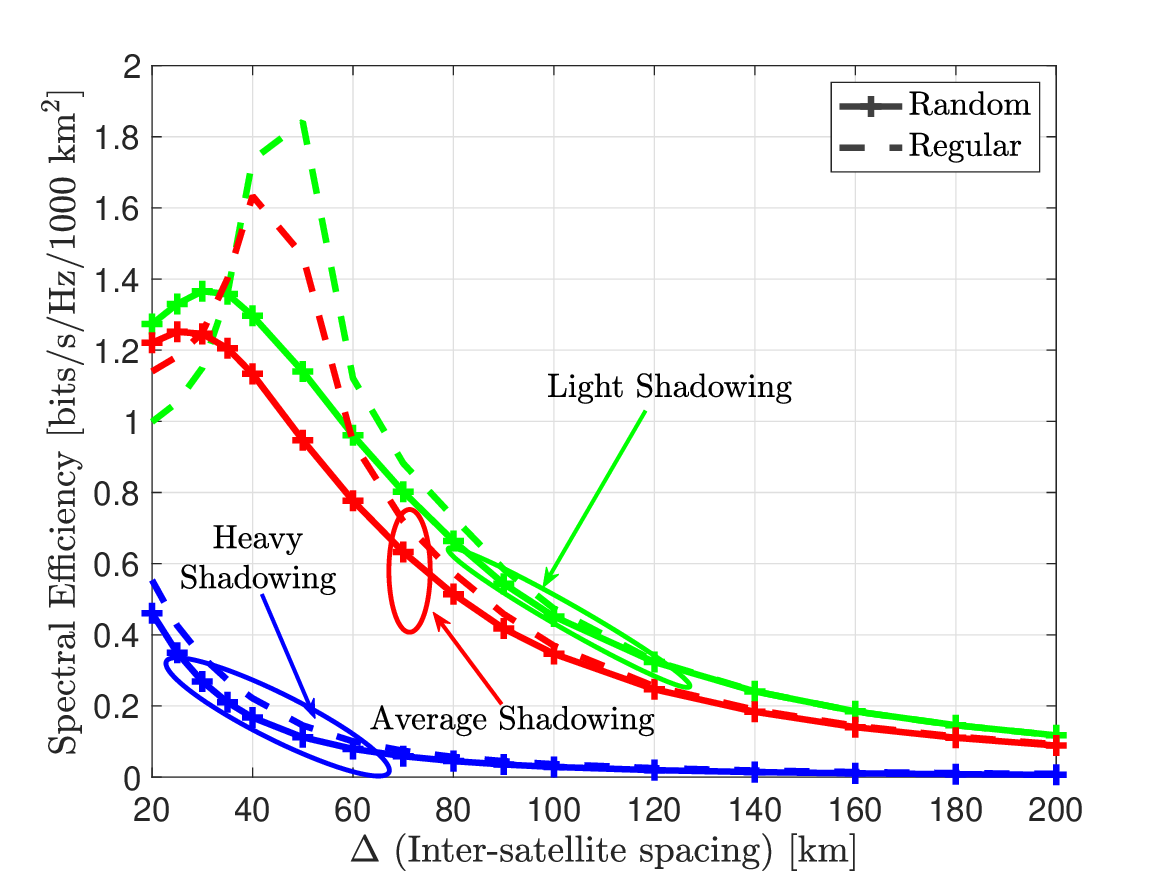}
            \caption{Spectral efficiency versus inter-satellite spacing for random and regular configurations under light, average and heavy shadowing scenarios when $(B_{\rm{sat}},B_{\rm{gs}}) = (5^\circ, 10^\circ)$,  $\psdmax h^{-\alpha}/\sigma^2 = 10~$dB.}
            \label{fig:spec_eff_single_chan_SR}
        \end{figure}

        Fig.~\ref{fig:single_channel_random_reg_various_Bsat_Bgs} compares spectral efficiencies versus inter-satellite spacing corresponding to randomly and regularly-placed satellite networks \co{along with the continuous approximation \textcolor{black}{\eqref{eq:Rcont}}} for the single channel network. Results are shown for three sets of beamwidths $(B_{\rm{sat}},B_{\rm{gs}})\in\{(5^\circ, 10^\circ), (10^\circ, 15^\circ), (20^\circ, 30^\circ)\}$. For the random configuration,  $\Delta$ denotes the expected inter-satellite spacing. The figure indicates that the spectral efficiency achieved by the regular configuration upper bounds that of the random configuration with the squared-distance association between satellites and terminals.%with the minimum-distance association between satellites and terminals. 
        In addition, for very small (or very large) values of $\Delta$, regular and random have similar spectral efficiencies. This is because as the density goes to zero, so does the interference at each terminal. Since the altitude of the satellites dominates the path-loss term for the direct-link, the regular and the random configurations behave similarly.  In addition, in the high-density regime the SINRs for both regular and random configurations approach the continuous approximation.

        %\mh{This is because as the density goes to zero, so does the interference at each terminal. Since the altitude of the satellites dominates the path-loss term for the direct-link, the regular and the random configuration behave similarly.  In addition, in the high-density regime the SINRs for both regular and random configurations approach the continuous approximation.}
        
       Fig.~\ref{fig:spec_eff_single_chan_SR} compares the spectral efficiencies of regular and random configurations with Shadowed Rician (SR) fading. In this setup, the gain of the link from satellite $i$ to terminal $k$ is multiplied by a random attenuation factor $\xi_{i,k}$, which represents the satellite-to-terminal fading effect. The values of $\xi_{i,k}$ are independent and identically distributed according to the Shadowed Rician distribution. In this model, the average power of the non-line-of-sight component is denoted by $2b$, the average power of the line-of-sight component is given by $\omega$, and the parameter $m$ characterizes the fading order~\cite{Abdi2003SSR, kim2024downlink}. The specific values of $b$, $m$, and $\omega$ corresponding to light, average, and heavy shadowing conditions can be found in~\cite{kim2024downlink}. The simulation is performed for a single channel network, where satellite to terminal associations are determined using the distance based rule. The results show that with this fading model the regular configuration continues to provide an upper bound on the spectral efficiency achieved by the random configuration.
% and that this bound appears to be significantly tighter when compared to the results without fading.

    \subsection{Multi-Channel Network}\label{sec:multivsreg}
        \co{In this section, a general satellite network described in Sec.~\ref{sec:SystemModel} is considered. Each satellite serves $N_B$ different terminals, where $N_B\in\{10, 20\}$. In other words, the number of terminals is equal to $N_B$ times the number of the satellites. To generate the satellite and terminal locations, points are first uniformly distributed over two parallel planar surfaces. These positions are then projected onto the surfaces of spheres with radii $r_e$ and $r_e + h$, corresponding to the terminals and satellites, respectively. Associations between satellites and terminals are again based on minimizing the squared-distance metric \eqref{eq:assoccost}.

        For subband allocation, a hexagonal grid is superimposed on the Voronoi regions formed by projecting the terminal positions onto the planar surface $\{[x, y, r_e] \mid x, y \in\mathbb{R}\}$, as illustrated in Fig.~\ref{fig:freq-reuse}.  This figure presents an example of a hexagonal frequency reuse pattern overlaid on the Voronoi regions of the terminals for $M = 4$, where hexagon centers are labeled with the digits 1 through 4. For example, the Voronoi region that includes the origin contains centers of two hexagons  and should be allocated subbands $\{1, 3\}$; the adjacent Voronoi region to the upper left is allocated subband $\{1\}$.  In the proposed model, a terminal with a Voronoi region that does not contain a hexagonal center is not served by any satellite. The primary reason for employing a hexagonal frequency reuse pattern on the ground plane is that a sphere cannot be perfectly tessellated using only hexagons. However, on a flat surface, hexagons provide the most efficient way to divide the area into regions of equal size with the smallest total perimeter~\cite{hales2002honeycombconjecture}.

        \begin{figure}
            \centering
            \includegraphics[width = \columnwidth]{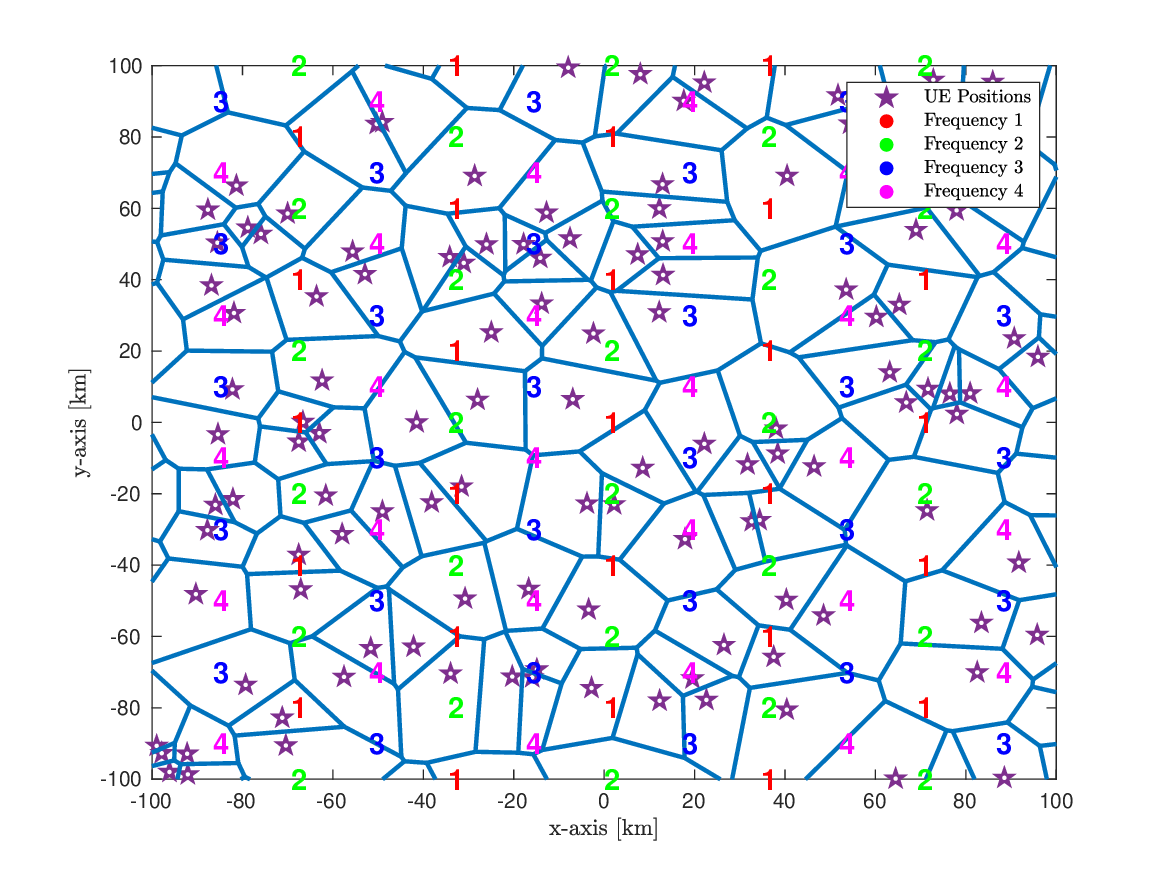}
            \caption{Voronoi regions of \co{projections of the randomly placed terminals to a plane} along with the centers of the hexagonal frequency cells with $M = 4$ subbands.}
            \label{fig:freq-reuse}
        \end{figure}

       Fig.~\ref{fig:spectralefficiency_random_simulator} presents the spectral efficiency as a function of inter-satellite spacing for multi-channel satellite networks. The transmission power is chosen such that $P_{\text{max}} h^{-\alpha} / (\aB \sigma^2)=8$ dB, and the power spectral density is set to $\psdmax = 10 P_{\text{max}} / \aB$. The beamwidths of the satellite and the terminal antennas are given by $(B_{\mathrm{sat}}, B_{\mathrm{gs}}) = (10^\circ, 20^\circ)$. Additionally, the maximum spectral efficiency attained across all values of $\Delta$ under the regular configuration is indicated as a horizontal line. 
        The plots include results for $M = 1$, 4, and 12 subbands. For each value of $\Delta$, the optimal co-channel distance within the hexagonal frequency reuse pattern is determined through numerical optimization. Increasing the number of subbands leads to a significant improvement in spectral efficiency when $\Delta$ is small, which corresponds to a dense network deployment. In particular, as the number of subbands increases, the performance approaches the maximum spectral efficiency of the regular configuration in the dense regime where $\Delta$ is small. However, this gain becomes negligible once $\Delta$ exceeds approximately 330 km.}
        
        % Furthermore, Fig.~\ref{fig:spectralefficiency_random_simulator} demonstrates the existence of an optimal spacing value $\Delta^*$ that maximizes spectral efficiency. For spacing values smaller than $\Delta^*$, interference becomes the limiting factor in network performance. .}

        % Table~\ref{tab:spec_eff_vs_num_freq} shows  spectral efficiency as a function of the number of subbands for the randomly placed network. Possible values for the number of subbands, $M$, are 
        % in the form of $m^2 +mn+n^2$ as discussed for $m, n \in\mathbb{N}$. Here we assume that $\psdmax$ is not a binding constraint. We consider the power-limited scenario $P_{\max} h^{-\alpha}/B\sigma^2 = 8$~dB with $(B_{\rm{sat}},B_{\rm{gs}}) = (10^\circ, 20^\circ)$ and two different satellite spacings: $\Delta=200$ km and $\Delta = 50~$km. The former is intended to correspond roughly to current satellite spacings, and the latter corresponds to a much denser network. The results show that the spectral efficiency does not change with respect to the number of subbands when $\Delta=200$ km. When $\Delta=50$ km, the spectral efficiency increases with the number of subbands up to $M=12$, and then remain basically constant, limited only by the SNR constraint rather than interference.
        % \footnote{In practice, the probability of an in-line interference event is also an important performance metric, so that additional bandwidth may be needed to ensure that this is sufficiently small.}

        \begin{figure}
            \centering
            \includegraphics[width = \columnwidth]{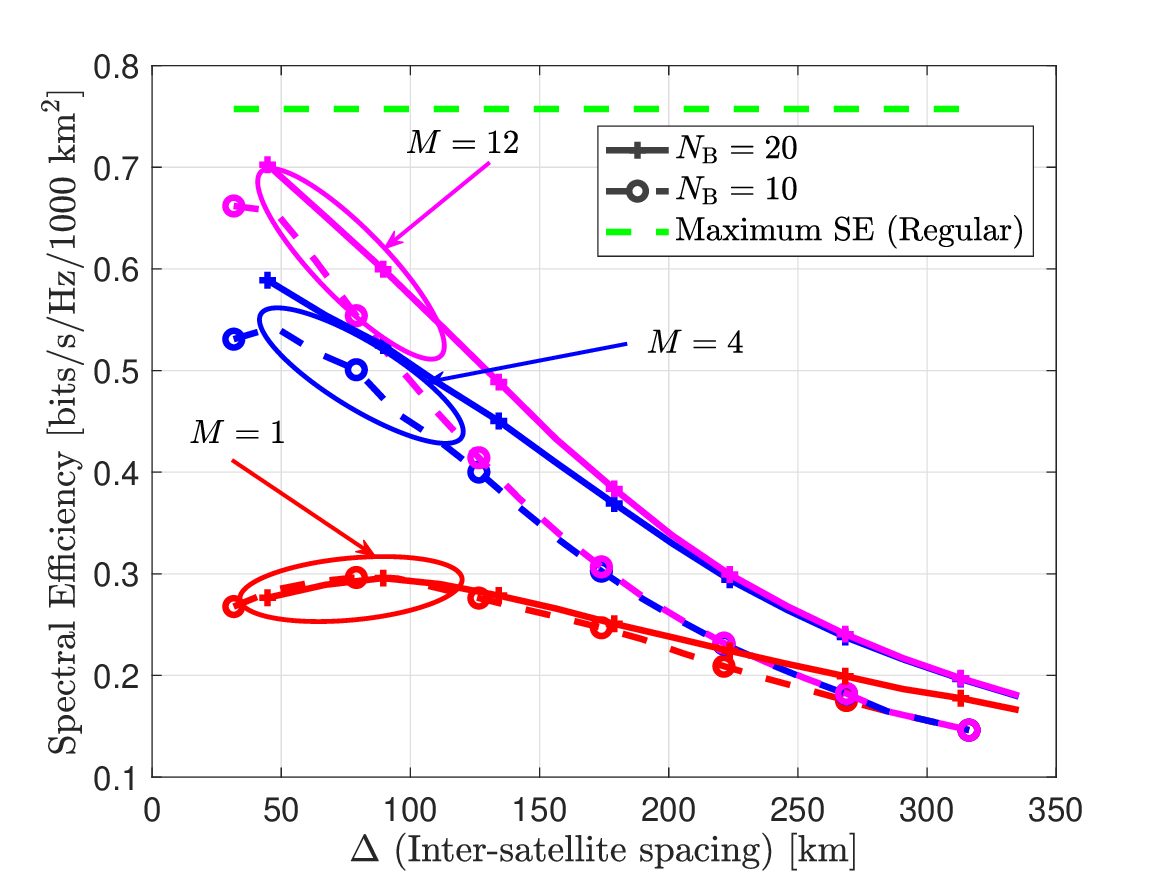}
            \caption{Spectral efficiency versus inter-satellite spacing with random placements and multiple channels. Also shown is the upper bound for the regular configuration. Parameters are $P_{\text{max}} h^{-\alpha} / (\aB \sigma^2) = 8$ dB,  $\psdmax = 10 P_{\text{max}}/\aB$, $(B_{\rm{sat}},B_{\rm{gs}}) = (10^\circ, 20^\circ)$,  number of subbands $M = 1, 4, 12$ 
            and each satellite serves $N_B = $ 10 or $N_B = $ 20 terminals.}
            \label{fig:spectralefficiency_random_simulator}
        \end{figure}

        % \begin{table}
        %     \begin{center}
        %         \begin{tabular}{| l ||c | c| c| c| c|} 
        %         \hline
        %         $M$ & 1  & 4 & 7 & 12 & 19 \\
        %          \hline
        %          $\Delta = 50$ km & 0.286  & 0.363 & 0.371 & 0.376 & 0.376 \\
        %          \hline
        %         $\Delta = 200$ km  & 0.072 & 0.072 & 0.072 & 0.072 & 0.072 \\ %[1ex] 
        %          \hline
        %         \end{tabular}
        %     \end{center}
        %      \caption{Spectral efficiency [bits/sec/Hz/1000 km$^2$] for different numbers of subbands ($M$) and spacings ($\Delta$).}
        %      \label{tab:spec_eff_vs_num_freq}
        % \end{table}

    \subsection{Shuffling Strategy for the Single Channel Network}
    
        \begin{figure}
            \centering
            \includegraphics[width = \columnwidth]{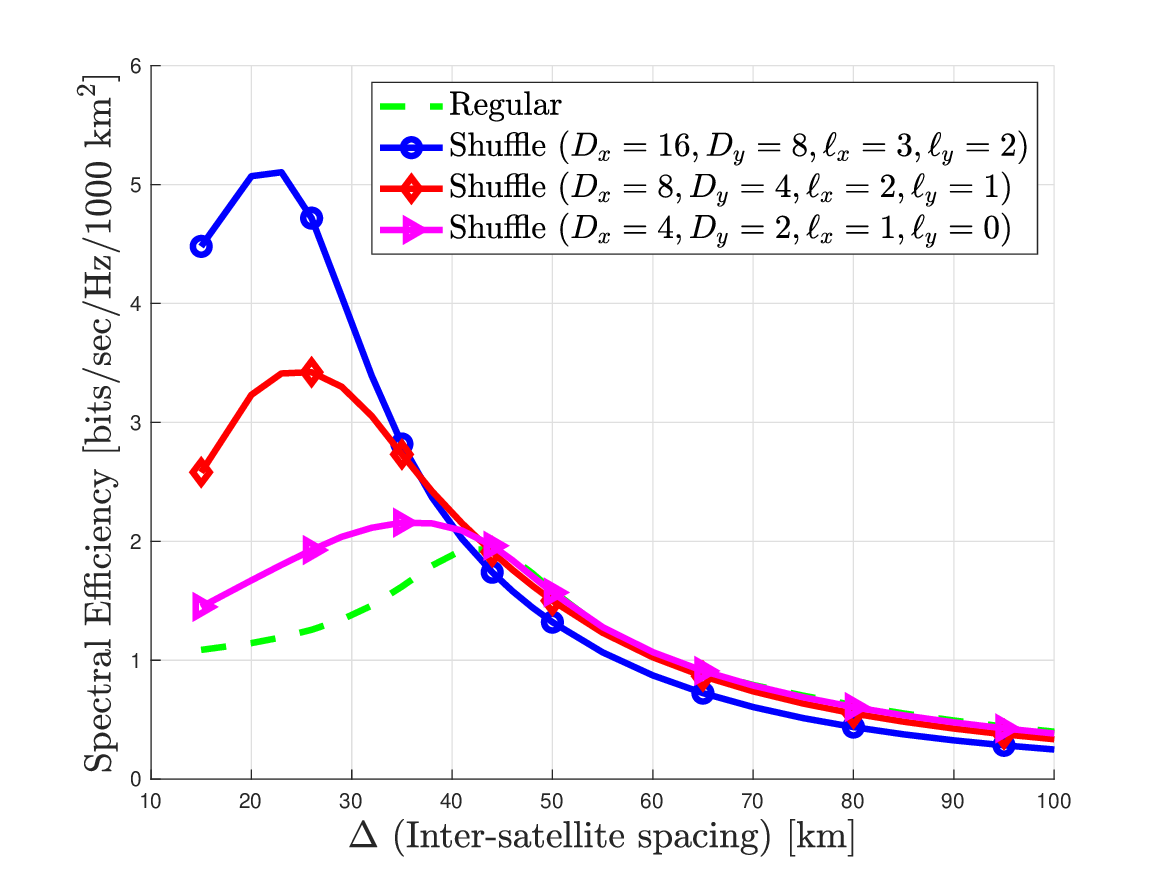}
            \caption{Spectral efficiency versus inter-satellite spacing for both minimum-distance and shuffled association strategies in a regularly placed network. Parameters: $(B_{\mathrm{sat}}, B_{\mathrm{gs}}) = (5^\circ, 10^\circ)$, $\psdmax h^{-\alpha} / \sigma^2 = 10$~dB, and $(D_x, D_y, \ell_x, \ell_y) \in \{(16, 8, 3, 2), (8, 4, 2, 1), (4, 2, 1, 0)\}$.}
            \label{fig:shufflevsreg}
        \end{figure}

    \co{
       The spectral efficiencies achieved by the shuffling strategy and the distance based association method under the regular configuration are shown in Fig.~\ref{fig:shufflevsreg}. Once the association is determined, the positions of the satellites and terminals are projected onto the surfaces of spheres with radii $r_e$ and $r_e + h$, respectively. The direct link gain is set to $\psdmax h^{-\alpha} / \sigma^2 = 10$~dB. The beamwidths of the satellite and the terminal antennas are selected as $(B_{\mathrm{sat}}, B_{\mathrm{gs}}) = (5^\circ, 10^\circ)$. 

        For the shuffling strategy, three sets of parameters are considered: $(D_x, D_y, \ell_x, \ell_y) \in \{(16, 8, 3, 2), (8, 4, 2, 1), (4, 2, 1, 0)\}$. Fig.~\ref{fig:shufflevsreg} illustrates that in the high-density regime, increasing the shuffling block sizes $(D_x, D_y)$ along with the number of shuffling rounds $(\ell_x, \ell_y)$ leads to a noticeable improvement in spectral efficiency. However, as the inter-satellite spacing $\Delta$ becomes larger, this performance gain gradually decreases. In fact, when $\Delta$ is sufficiently large, applying the shuffling strategy results in a reduction in spectral efficiency compared to the baseline configuration.

       }

\section{Conclusions}\label{sec:Conc}
    \co{In general, determining the maximum spectral efficiency of a LEO satellite network while jointly optimizing satellite and terminal locations, association strategies, and frequency allocations is a highly complex task. To address this, we introduced the regular configuration of satellites and terminals as a tractable benchmark for estimating meaningful upper bounds on the sum spectral efficiency. This was achieved through a reduction to a virtual single-channel network model. Our numerical evaluations indicate that the spectral efficiency attained under the regular configuration provides a meaningful upper bound for randomly generated multi-channel satellite networks, particularly when distance-based association rules are employed.  It is also worth noting that the single-channel network reduction is still valid if we consider multi-altitude satellite networks and/or introduce fading into our model. For the multi-altitude network scenario, we can approximate our geometry with multiple parallel planes. 
    
    Furthermore, our results reveal that substantial gains in spectral efficiency can be achieved by modifying satellite-to-terminal associations. While we proposed a  shuffling strategy to improve associations, the determination of an optimal association policy remains an open problem for future research. Additionally, this work assumes full coordination among satellites, a condition that may not hold in practice, especially in competitive environments with multiple service providers. Extending the analysis to such scenarios, including distributed or contention-based resource allocation without global coordination, represents a promising direction for future work.}

\appendices

\section{Proof of Proposition 2}\label{sec:proofprop2}
\begin{proof}
Let $\mathbf P=(P_1,\dots,P_N)\in[0,\psdmax]^N$ collect the transmit powers and define the network sum spectral efficiency
\begin{align}\label{eq:Phidef}
\Phi(\mathbf P)\triangleq\sum_{k=1}^{N}\log_2\!\left(1+\mathrm{SINR}_k\right),
\end{align}
with $\mathrm{SINR}_k$ as in~\eqref{eq:SINR_k}. For each terminal $k$, denote the interference-plus-noise and the total received power by
\begin{align}
\nu_k&\triangleq\sum_{i\in\mathcal I_k}P_i\,\Omega_{i,k}+\sigma^2,\label{eq:nuk}\\
\tau_k&\triangleq \nu_k+P_k\,\Omega_{k,k},\label{eq:tauk}
\end{align}
respectively, so that $1+\mathrm{SINR}_k=\tau_k/\nu_k$. Substituting this into~\eqref{eq:Phidef},
\begin{align}\label{eq:Philog}
\Phi(\mathbf P)=\frac{1}{\log 2}\sum_{k=1}^{N}\bigl(\log \tau_k-\log \nu_k\bigr).
\end{align}
Because $k\notin\mathcal I_k$, the power $P_k$ does not appear in $\nu_k$; thus $\nu_k$ varies with $P_j$ only through the interferers $j\in\mathcal I_k$, whereas $\tau_k$ varies in addition through $P_k$:
\begin{align}
\frac{\partial \nu_k}{\partial P_j}&=
\begin{cases}\Omega_{j,k}, & j\in\mathcal I_k,\\[2pt] 0,&\text{otherwise},\end{cases}\label{eq:dnuk}\\[4pt]
\frac{\partial \tau_k}{\partial P_j}&=
\begin{cases}\Omega_{k,k}, & j=k,\\[2pt] \Omega_{j,k}, & j\in\mathcal I_k,\\[2pt] 0,&\text{otherwise}.\end{cases}\label{eq:dtauk}
\end{align}
Differentiating~\eqref{eq:Philog} and using~\eqref{eq:dnuk}--\eqref{eq:dtauk},
\begin{align}
\frac{\partial \Phi}{\partial P_j}
&=\frac{1}{\log 2}\sum_{k=1}^{N}\left(\frac{1}{\tau_k}\frac{\partial \tau_k}{\partial P_j}-\frac{1}{\nu_k}\frac{\partial \nu_k}{\partial P_j}\right)\nonumber\\
&=\frac{1}{\log 2}\left[\frac{\Omega_{j,j}}{\tau_j}+\!\!\sum_{k:\,j\in\mathcal I_k}\!\!\Omega_{j,k}\left(\frac{1}{\tau_k}-\frac{1}{\nu_k}\right)\right],
\end{align}
where the first term is the $k=j$ contribution and the sum collects the terminals $k$ for which satellite $j$ is an interferer. Since $\tau_k-\nu_k=P_k\Omega_{k,k}$ by~\eqref{eq:tauk}, we have $1/\tau_k-1/\nu_k=-P_k\Omega_{k,k}/(\nu_k\tau_k)$, and therefore
\begin{align}\label{eq:grad_general}
\frac{\partial \Phi}{\partial P_j}=\frac{1}{\log 2}\Bigg[\frac{\Omega_{j,j}}{\tau_j}-\!\!\sum_{k:\,j\in\mathcal I_k}\!\!\frac{P_k\,\Omega_{k,k}\,\Omega_{j,k}}{\nu_k\,\tau_k}\Bigg].
\end{align}
We now evaluate~\eqref{eq:grad_general} at the full-power point $\mathbf P=\psdmax\mathbf 1_N$, where $\mathbf1_N$ is  an $N$-dimensional all ones vector. Condition~1 gives $\Omega_{k,k}=\Omega_{1,1}$ for all $k$, while the incoming-interference part of condition~2, $\sum_{i\in\mathcal I_k}\Omega_{i,k}=\bar I$, makes the denominators in~\eqref{eq:nuk}--\eqref{eq:tauk} identical across terminals:
\begin{align}\label{eq:nutaudef}
\nu_k=\nu\triangleq\psdmax\bar I+\sigma^2,\qquad
\tau_k=\tau\triangleq \nu+\psdmax\,\Omega_{1,1}.
\end{align}
The outgoing-interference part of condition~2 gives $\sum_{k:\,j\in\mathcal I_k}\Omega_{j,k}=\bar I$ for every satellite $j$. 
Substituting~\eqref{eq:nutaudef} and these identities into~\eqref{eq:grad_general},
\begin{align}
\frac{\partial \Phi}{\partial P_j}\bigg|_{\psdmax\mathbf 1}
&=\frac{1}{\log 2}\left[\frac{\Omega_{1,1}}{\tau}-\frac{\psdmax\,\Omega_{1,1}}{\nu\,\tau}\!\!\sum_{k:\,j\in\mathcal I_k}\!\!\Omega_{j,k}\right]\nonumber\\
&=\frac{\Omega_{1,1}}{\tau \log 2}\left(1-\frac{\psdmax\bar I}{\nu}\right)
=\frac{\Omega_{1,1}}{\tau \log 2}\cdot\frac{\sigma^2}{\nu}>0,\label{eq:grad_pos}
\end{align}
where the last step uses $\nu-\psdmax\bar I=\sigma^2$ from~\eqref{eq:nutaudef}.
Finally, $\psdmax\mathbf 1_N$ is the corner of the feasible box $[0,\psdmax]^N$ at which all constraints are active, so every feasible direction $\mathbf u \triangleq [u_1 \ldots u_N]^\intercal$ satisfies $u_j\le 0$ for all $1\le j\le N$. By~\eqref{eq:grad_pos} each partial derivative at $\psdmax\mathbf 1_N$ equals a common constant $c>0$, hence
\begin{align}
\nabla\Phi(\psdmax\mathbf 1_N)^{\!\top}\mathbf u
= c\sum_{j=1}^{N}u_j
= -\,c\sum_{j=1}^{N}\lvert u_j\rvert
\le -\,c\,\lVert\mathbf u\rVert .
\end{align}
Since $\Phi$ is continuously differentiable,
\begin{align}
    \Phi(\psdmax\mathbf 1_N+\mathbf u)-\Phi(\psdmax\mathbf 1_N)&=\nabla\Phi(\psdmax\mathbf 1_N)^{\!\top}\mathbf u+o(\lVert\mathbf u\rVert)\\
    &\le -c\lVert\mathbf u\rVert+o(\lVert\mathbf u\rVert),
\end{align}
which is strictly negative for every feasible $\mathbf u\neq\mathbf 0$ with $\lVert\mathbf u\rVert$ sufficiently small. Hence $\psdmax\mathbf 1_N$ is a strict local maximizer of $\Phi$ over $[0,\psdmax]^N$.
% Finally, we upgrade this first-order property to strict local optimality. Each partial derivative $\partial\Phi/\partial P_j$ is continuous in $\mathbf P$ on $[0,\psdmax]^N$, since the denominators $\tau_j$ and $\nu_k\tau_k$ in~\eqref{eq:grad_general} are bounded below by $\sigma^2$ and $\sigma^4$, respectively. By~\eqref{eq:grad_pos}, all $N$ partial derivatives are strictly positive at $\mathbf P=\psdmax\mathbf 1$, and a finite intersection of open sets is open, so there exists $r>0$ with
% \begin{align}\label{eq:grad_pos_nbhd}
% \frac{\partial \Phi}{\partial P_j}(\mathbf P)>0,\quad \forall j,\ \ \forall\,\mathbf P\in\mathcal N_r\triangleq[\psdmax-r,\psdmax]^N.
% \end{align}
% Take any feasible $\mathbf P\in\mathcal N_r$ with $\mathbf P\neq\psdmax\mathbf 1$, and raise its coordinates to $\psdmax$ one at a time. Each step increases a single coordinate $P_j$ while the iterate stays in $\mathcal N_r$; by~\eqref{eq:grad_pos_nbhd}, $\Phi$ is nondecreasing along it and strictly increasing whenever $P_j<\psdmax$. Since $\mathbf P\neq\psdmax\mathbf 1$, at least one coordinate is strictly raised, so $\Phi(\psdmax\mathbf 1)>\Phi(\mathbf P)$. Hence $\psdmax\mathbf 1$ is a strict local maximizer of $\Phi$ over $[0,\psdmax]^N$.
\end{proof}
    \section{Proof of Proposition 3}\label{sec:proofprop3}
     Recall from \eqref{eq:reg_angle} the slant range
    $D^{\mathrm{reg}}_{i,j}(\Delta)=\lVert s^{\mathrm{reg}}_{i,j}-g^{\mathrm{reg}}_{0,0}\rVert$
    between the satellite at $s^{\mathrm{reg}}_{i,j}$ and the reference terminal at
    $g^{\mathrm{reg}}_{0,0}=[0,0,r_e]$. It should be noted that
    \begin{align}
    	\lim_{\Delta\to 0} D_{2i, 2j}^{\text{reg}}(\Delta) = \lim_{\Delta\to 0} D_{2i+1, 2j+1}^{\text{reg}}(\Delta) = h,
    \end{align}
    and
    \begin{align}
    	\lim_{\Delta\to 0} \theta_{2i, 2j}^{\text{reg}}(\Delta) = \lim_{\Delta\to 0} \theta_{2i+1, 2j+1}^{\text{reg}}(\Delta) = 0.
    \end{align}
    Thus, it is evident that
    \begin{align}
        \lim_{\Delta\to 0} \eta(\psdmax, \Delta) = \infty.
    \end{align}
    Then, by  L'Hôpital's rule,
    \begin{align}
    	&\lim_{\Delta\to 0}  R^{\text{reg}}(\psdmax, \Delta) = \lim_{\Delta\to 0} \frac {\log_2\left(1 + \frac{\gamma(\psdmax) } { \eta(\psdmax, \Delta) + 1}\right)}{\Delta^2 \sqrt{3}/2}  \nonumber\\
    	&=\lim_{\Delta\to 0} \Bigg( \frac{-\gamma(\psdmax)   \eta'(\psdmax, \Delta)}{\Delta\sqrt{3} \log{2}  \left(\eta(\psdmax, \Delta) + 1\right)^2}  \frac{1}{1 + \frac{\gamma (\psdmax) } { \eta(\psdmax, \Delta) + 1} } \Bigg),
    \end{align}
    where $\eta'(\psdmax, \Delta)$ denotes the first derivative of $\eta(\psdmax, \Delta)$ with respect to $\Delta$. Since we have
    \begin{align}
    	\lim_{\Delta\to 0} \frac{\eta(\psdmax, \Delta) + 1}{\gamma(\psdmax)  + \eta(\psdmax, \Delta) + 1 } = 1,
    \end{align}
    to show boundedness of $\lim_{\Delta\to 0}  R^{\text{reg}}(\psdmax,\Delta)$, it is sufficient to prove the following limit
    \begin{align}
    	\lim_{\Delta\to 0} \frac{\eta'(\psdmax, \Delta)}{\Delta \left(\eta(\psdmax, \Delta) + 1\right)^2}
    \end{align}
    is bounded. By  L'Hôpital's rule, one can see that
    \begin{align}
    	\lim_{\Delta\to 0} \frac{\eta'(\psdmax, \Delta)}{\Delta \left(\eta(\psdmax, \Delta) + 1\right)^2}  = \lim_{\Delta\to 0} \frac{-2}{\left(\eta(\psdmax, \Delta) + 1\right)\Delta^2}
    \end{align}
    Thus,  if we can prove that
    \begin{align}
    	\lim_{\Delta\to 0} \eta(\psdmax, \Delta) \Delta^2 > 0,
    \end{align}
    we are done.
    
    For any given $\epsilon >0$, let $\widetilde{D}>0$ be such that the aggregate interference at $(0,0)$ due to the satellites having distance larger than $\widetilde{D}$ is less than $\epsilon$. There is always such $\widetilde{D}$, because we know $\sum_{n} n^{-\alpha}$ is a convergent series for $\alpha > 1$. Clearly, the value of $\widetilde{D}$ depends on the beamwidths of the satellites and terminals. Accordingly, let $D \triangleq  \sqrt{\widetilde{D}^2-h^2}$. 
    
    Let us define  $I_0$ and  $I_1$ as follows
    \begin{align}
    	I_0 &\triangleq  \frac{\psdmax}{\sigma^2} \sum_{(i,j) \neq (0,0)} \left(D_{2i, 2j}^{\text{reg}}(\Delta)\right)^{-\alpha} W\left(\theta^{\text{reg}}_{2i,2j}(\Delta)\right), \\
     I_1&\triangleq  \frac{\psdmax}{\sigma^2} \sum_{(i,j)} \left(D_{2i+1, 2j+1}^{\text{reg}}(\Delta)\right)^{-\alpha} W\left(\theta^{\text{reg}}_{2i+1,2j+1}(\Delta)\right),
    \end{align}
    where $W(\theta)\triangleq w_s(\theta)w_g(\theta)$. It is clear that $\eta(\psdmax, \Delta)  = I_0+ I_1$. If $(2i, 2j)^\mathrm{th}$ link is leading a non-negligible interference at the terminal located $(0,0)$, then we must have
    \begin{align}
    	i^2 + 3j^2 \leq \frac{D^2}{\Delta^2}\, \text{and}\, (i, j) \neq (0, 0).
    \end{align}
    
     The number of such links, $T_0$, can be easily lower bounded as
    \begin{align}
    	T_0 &= \sum_{i=0}^{T} \left(2\floor[\Bigg]{\sqrt{\frac{T^2-i^2}{3}}} -1\right) \\
        &\ge  \frac{2}{\sqrt{3}}\sum_{i=0}^{T} \sqrt{T^2-i^2} - (2T + 1)\nonumber \\
    	&\geq \frac{2}{\sqrt{3}}\int_{0}^{T} \sqrt{T^2-x^2} \, dx - (2T + 1)\\
     &= \frac{2\pi}{4\sqrt{3}}T^2 -(2T + 1),
    \end{align}
    where $T = D/\Delta$. Via similar steps, the number of links, $T_1$, causing a non-negligible interference term in the summation $I_1$, can be lower bounded by the same number.
    
    Let us denote $\widetilde{I}$ as the interference term obtained when the distance between the satellite and the terminal is equal to $\widetilde{D}$. Thus, we can write that
    \begin{align}
    	\eta(\psdmax, \Delta) \Delta^2 &> \frac{\pi T^2 \Delta^2}{\sqrt{3}}  \widetilde{I} - (4T + 2) \widetilde{I}\Delta^2 \\
     &= \frac{\pi D^2} {\sqrt{3}}\widetilde{I} - \widetilde{I} (4 D\Delta + 2\Delta^2).
    \end{align}
    In other words, we show that 
    \begin{align}
    	\lim_{\Delta\to 0} \eta(\psdmax, \Delta) \Delta^2  > \frac{\pi D^2} {\sqrt{3}}\widetilde{I} > 0.
    \end{align}
    Hence, we are done. \hfill $\square$

\section{Proof of Proposition 4}\label{sec:proofprop4}
    \begin{lemma}\label{lemma:1}
        Assuming $w_s(\theta) w_g(\theta)$ is a non-increasing function of $\theta\in[0, \pi/2]$, $ (\hat{x}, \hat{y}) = \argmax_{x, y} \mathrm{SINR}(x,y)$ implies that $ -\Delta/2 \leq \hat{x} \leq \Delta/2$ and $ -\Delta\sqrt{3}/2 \leq \hat{y} \leq \Delta\sqrt{3}/2.$
    \end{lemma}
    \begin{proof}
    Let $\mathrm{SNR}(x, y)$ and $\mathrm{INR}(x, y)$ be defined as follows
    \begin{align}
        \mathrm{SNR}(x, y) &\triangleq \frac{\psdmax}{\sigma^2} D_{0,0}^{-\alpha}(x, y) W_{0,0}(x,y), \\
        \mathrm{INR}(x, y) &\triangleq \frac{\psdmax}{\sigma^2} \sum_{ (i, j) \in \mathcal{I}} D_{i,j}^{-\alpha}(x, y) W_{i,j}(x, y),
    \end{align}
        Take any $x, y\in\mathbb{R}$ with  $x \geq \Delta/2$.   Note that
    \begin{align}
        \frac{\mathrm{SNR}(x-\Delta,y)}{\mathrm{SNR}(x, y)} &= \frac{D_{0,0}^{-\alpha}(x-\Delta,y) W_{0,0}(x-\Delta, y)}{D_{0,0}^{-\alpha}(x,y) W_{0,0}(x, y)} \nonumber \\
        &= \frac{D_{2,0}^{-\alpha}(x,y) W_{2,0}(x, y)}{D_{0,0}^{-\alpha}(x,y) W_{0,0}(x, y)} \label{eq:SNR_00_20},
    \end{align}
    and
    \begin{align}
		&\mathrm{INR}(x,y)-\mathrm{INR}(x-\Delta,y)=   \nonumber\\
		&\frac{\psdmax}{\sigma^2} \left(D_{2,0}^{-\alpha}(x,y) W_{2,0} (x, y) - D_{0,0}^{-\alpha}(x,y) W_{0,0} (x, y) \right) \label{eq:INR_00_20}.
	\end{align}
    It is possible to write that, 
	\begin{align}
		D_{2,0}(x,y) & = \sqrt{(x-\Delta)^2 + y^2 + h^2} \nonumber\\
		&\leq  \sqrt{x^2 + y^2 + h^2} = D_{0,0}(x, y), \label{ineq:D_00_D_200}
	\end{align}
    and 
    \begin{align}
    \theta_{2, 0}(x,y) &= \cos^{-1}\left(h/ \sqrt{(x-\Delta)^2 + y^2 + h^2}\right)\nonumber\\
    &\leq \cos^{-1}\left(h/ \sqrt{x^2 + y^2 + h^2}\right) = \theta_{0,0}(x,y).
    \end{align}
    We can argue that
    \begin{align}
          \theta_{2,0}(x,y) \leq \theta_{0,0}(x,y) \implies W_{0,0}(x,y) \leq W_{2,0}(x,y), \label{ineq:W_00_W_200}
    \end{align}
   as $w_s(\cdot)w_g(\cdot)$ is non-increasing in $[0, \pi/2]$. 
   
   As a consequence of \eqref{eq:SNR_00_20}-\eqref{ineq:W_00_W_200}, we can conclude that $\mathrm{SINR}(x-\Delta,y)\geq \mathrm{SINR}(x, y)$.  Similarly, one can show that for any $x\leq -\Delta/2$, $\mathrm{SINR}(x +\Delta,y)\geq \mathrm{SINR}(x, y)$. In addition, we can prove that for any $y \geq \Delta\sqrt{3}/2$,  we have $\mathrm{SINR}(x,y-\Delta\sqrt{3}/2)\geq \mathrm{SINR}(x, y)$ and  for any $y \leq -\Delta\sqrt{3}/2$,  we have $\mathrm{SINR}(x,y+\Delta\sqrt{3}/2)\geq \mathrm{SINR}(x, y)$. Thus, we conclude that we can focus on the intervals  $-\Delta/2 \leq x \leq \Delta/2$ and $-\Delta\sqrt{3}/2 \leq y \leq \Delta\sqrt{3}/2$.
    \end{proof}

    \begin{lemma}\label{lemma:2}
         For any $i=j\Mod{2}$ and $(i,j)\neq(0,0)$, define
             \begin{align}
            \rho_{i, j}(x, y) \triangleq \left(\frac{D_{0,0}(x,y)}{D_{i,j}(x,y)}\right)^2,
        \end{align}
         Then, $\rho_{i, j}(x, y)$	is convex in both $x$ and $y$ for $-\Delta/2 \leq x \leq \Delta/2$ and $-\Delta\sqrt{3}/2 \leq y \leq \Delta\sqrt{3}/2$.
        \end{lemma}
        \begin{proof}
        We will prove the convexity of $\rho_{i,j}(x,y)$ when $i$ and $j$ are even.  The other case, i.e., $i = j = 1\Mod{2}$ can be proven similarly.
        
        Let $a = i/2,~b = j/2,~\ell =  y^2 + h^2$, and $m = (y-\Delta\sqrt{3}b)^2 + h^2$. It should be noted that
        \begin{align}
            \ell-m = (2y-\Delta\sqrt{3}b) \Delta\sqrt{3}b\leq 0, \label{ineq_l_m}
        \end{align}
        for any $b\in\mathbb{Z}$.
        
         After some algebraic manipulations, it is easily verified that in order to show the convexity of $\rho_{i, j}(x,y)$ in $x$,  the following inequality must hold
        \begin{align}
        &\left(x-\Delta a\right)^2 \left(\Delta^2a^2 + 2\Delta a x + 3\ell \right) \nonumber \\
            &+ m (m-\ell + 2\Delta^2a^2- 3x^2) \geq 0  \label{eq:ineq1}
        \end{align}
        \begin{enumerate}
            \item Case I: ($a = 0$)
            After rearranging~\eqref{eq:ineq1}, we must prove the following inequality
            \begin{align}
                (m-\ell)(m-3x^2)\geq 0\label{eq:ineq01}.
            \end{align}
            Then, ~\eqref{eq:ineq01} reduces to $m \geq 3x^2$ due to \eqref{ineq_l_m}.  Since $(i, j) \neq (0, 0)$, $b$ cannot be equal to $0$. In other words, $\abs{b}\geq 1$. Note that
             \begin{align}
                m = (y-\Delta\sqrt{3}b)^2 + h^2 > (y-\Delta\sqrt{3}b)^2 \geq \frac{3\Delta^2}{4} \geq 3x^2
            \end{align}
            as we desired to prove. 
            \item Case II: $(a \neq 0)$
            
            First, note that
            \begin{align}
                \Delta^2a^2 + 2\Delta a x  \geq 0, \label{ineq_CaseII_1}
            \end{align}
            for any $a\in\mathbb{Z}$. 
            
            Secondly, as $x^2 \leq \Delta^2/4$
            \begin{align}
                2\Delta^2a^2- 3x^2 \geq 0,\label{ineq_CaseII_2}
            \end{align}
            for any $a\neq 0$. Therefore, \eqref{eq:ineq1} is true by combining \eqref{ineq_l_m}, \eqref{ineq_CaseII_1}, and \eqref{ineq_CaseII_2}.

            % Therefore, \eqref{eq:ineq01} is true by combining \eqref{ineq_CaseII_1} and \eqref{ineq_CaseII_2}.
            
             Thus, we prove the convexity of $\rho_{i,j}(x,y)$ in $x$ when $i = j =  0\Mod{2}$. The proof of convexity in $y$ can be done similarly.
        \end{enumerate}
    \end{proof}
    As a consequence of Lemma~\ref{lemma:1}, we can restrict our attention to the intervals $-\Delta/2 \leq x \leq \Delta/2$ and $-\Delta\sqrt{3}/2 \leq y \leq \Delta\sqrt{3}/2$.  We also assume $x, y \geq 0$. The other three cases regarding the signs of $x$ and $y$ can be considered similarly.
    We can write 
	\begin{align}
		&\frac{1}{\mathrm{SINR}(x, y)} = \sum_{\left(i, j\right) \in \mathcal{I}} \frac{D_{i,j}^{-\alpha}(x, y) W_{i, j}(x, y)}{D_{0,0}^{-\alpha}(x, y) W_{0, 0}(x, y)} + \frac{1}{\mathrm{SNR}(x,y)} \nonumber \\
		& =  \sum_{\left(i, j\right) \in \mathcal{I}}\frac{ W_{i, j}(x, y) }{ W_{0, 0}(x, y)}\rho^{\alpha/2}_{i, j}(x, y) + \frac{1}{\mathrm{SNR}(x,y)}.
	\end{align}
    % Due to Lemma~\ref{lemma:2}, $\rho_{i,j}(x,y)$ is convex in $x$ and $y$. Hence,  for $\alpha > 2$, we can easily note that $\rho^{\alpha/2}_{i,j}(x,y)$ is convex in $x$ and $y$ as well.
    Due to Lemma~\ref{lemma:2}, $\rho_{i,j}(x,y)$ is convex and nonnegative in $x$ and $y$. Since $t\mapsto t^{\alpha/2}$ is convex and nondecreasing on $[0,\infty)$ for $\alpha\ge 2$, the composition $\rho^{\alpha/2}_{i,j}(x,y)$ is convex in $x$ and $y$ as well. Then we can argue for any $i, j \geq 1$ that
    \begin{align}
		\rho^{\alpha/2}_{2i, 2j}(x, y) + \rho^{\alpha/2}_{2i, 2j}(-x, y) \geq 2 \rho^{\alpha/2}_{2i, 2j}(0, y).
	\end{align}
    In addition, from the second condition in Proposition~\ref{Prop:regopt}, we can write
	\begin{align}
		\frac{W_{2i, 2j}(x, y) }{W_{0, 0}(x, y)} + \frac{W_{2i, 2j}(-x, y) }{W_{0, 0}(x, y)} \geq 2 \frac{W_{2i, 2j}(0, y)}{W_{0, 0}(0, y)},
	\end{align}
	It is also clear that
	\begin{align}
		\frac{W_{2i, 2j}(x, y) }{W_{0, 0}(x, y)} & \geq \frac{W_{2i, 2j}(-x, y) }{W_{0, 0}(x, y)} = \frac{W_{-2i, 2j}(x, y) }{W_{0, 0}(x, y)},  \\
  \rho_{2i, 2j}(x, y) &\geq \rho_{2i, 2j}(-x, y) = \rho_{-2i, 2j}(x, y) . 
	\end{align}
	Then, by using the rearrangement inequality, we reach that
	\begin{align}
		&\frac{W_{2i, 2j}(x, y) }{W_{0, 0}(x, y)}\rho^{\alpha/2}_{2i, 2j}(x, y) + \frac{W_{-2i, 2j}(x, y) }{W_{0, 0}(x, y)}\rho^{\alpha/2}_{-2i, 2j}(x, y)\\
  &\geq 2 \frac{W_{2i, 2j}(0, y)}{W_{0, 0}(0, y)}\rho^{\alpha/2}_{2i, 2j}(0, y).
	\end{align}
	Similarly, we have
	\begin{align}
		&\frac{W_{2i, -2j}(x, y) }{W_{0, 0}(x, y)}\rho^{\alpha/2}_{2i, -2j}(x, y) + \frac{W_{-2i, -2j}(x, y) }{W_{0, 0}(x, y)}\rho^{\alpha/2}_{-2i, -2j}(x, y) \\
  &\geq 2 \frac{W_{2i, -2j}(0, y)}{W_{0, 0}(0, y)}\rho^{\alpha/2}_{2i, -2j}(0, y) = 2 \frac{W_{2i, 2j}(0, -y)}{W_{0, 0}(0, -y)}\rho^{\alpha/2}_{2i, 2j}(0, -y).
	\end{align}
	By using convexity in $y$ and doing the same steps above, we can claim that 
	\begin{align}
		(0,0) = \argmin_{(x,y)\in\mathbb{R}^2}\sum_{\left(i, j\right) \neq (0,0)} \frac{D_{2i,2j}^{-\alpha}(x, y) W_{2i, 2j}(x, y)}{D_{0,0}^{-\alpha}(x, y) W_{0, 0}(x, y)}.
	\end{align}
	Similarly, we can prove that 
	\begin{align}
		(0, 0) = \argmin_{(x,y)\in\mathbb{R}^2} \sum_{i, j} \frac{ W_{2i+1, 2j+1}(x, y) }{ W_{0, 0}(x, y)}\rho_{2i+1, 2j+1}(x, y).
	\end{align}
	As $(0, 0)$ is the maximizer of $\mathrm{SNR}(x, y)$, we reach the desired conclusion. \hfill $\square$

    \bibliographystyle{IEEEtran}
    \bibliography{bibfile}
    
\end{document}